\documentclass[journal]{IEEEtran}
\usepackage[utf8]{inputenc}
\usepackage{amsmath, amssymb, amsthm}
\usepackage{enumerate}
\usepackage{graphicx}
\usepackage{xcolor}
\usepackage{booktabs}
\usepackage{subfigure}
\usepackage{subcaption}
\usepackage{lipsum}

\usepackage{algorithm, algorithmicx, algpseudocode}
\algrenewcommand\algorithmicrequire{\textbf{Input:}}
\algrenewcommand\algorithmicensure{\textbf{Output:}}
\algnewcommand{\BlackBox}[1]{%
    \begin{flushleft}
    \hspace{-.7cm}
    \textbf{Available Functions:}
    {\raggedright #1}
    \end{flushleft}
}
\algnewcommand{\Initialize}[1]{%
    \begin{flushleft}
    \hspace{-.7cm}
    \textbf{Initialize:}
    {\raggedright #1}
    \end{flushleft}
}
\algdef{SE}[DOWHILE]{Do}{doWhile}{\algorithmicdo}[1]{\algorithmicwhile\ #1}

\allowdisplaybreaks

\usepackage{hyperref}
\usepackage[capitalize,nameinlink]{cleveref}
\makeatletter
\newcommand{\ALG@lineautorefname}{Step}
\makeatother

\usepackage{tikz}
\usepackage{pgfplots}
\pgfplotsset{compat=newest,compat/show suggested version=false}
\usetikzlibrary {shapes,arrows,positioning}
\usetikzlibrary {arrows.meta} 
\usepackage{multirow}
\usepackage{color, colortbl}
\definecolor{Gray}{gray}{0.9}
\usepgfplotslibrary{fillbetween} 
\usepackage{xfp}
\usepackage{pgfplotstable}

\usepackage{centernot}

\usepackage{color, soul}

\usepackage{mdframed}
\mdfdefinestyle{MyFrame}{%
    linecolor=black,
    outerlinewidth=2pt,
    roundcorner=20pt,
    innertopmargin=\baselineskip,
    innerbottommargin=\baselineskip,
    innerrightmargin=10pt,
    innerleftmargin=10pt,
    backgroundcolor=gray!20!white
}

\newcommand{\CA}{{\mathcal A}}
\newcommand{\CB}{{\mathcal B}}
\newcommand{\CC}{{\mathcal C}}
\newcommand{\CD}{{\mathcal D}}
\newcommand{\CE}{{\mathcal E}}
\newcommand{\CF}{{\mathcal F}}

\newcommand{\CP}{{\mathcal P}}

\newcommand{\CX}{{\mathcal X}}
\newcommand{\CY}{{\mathcal Y}}

\newcommand{\fpkq}{\CF_{p}(k, q)}
\newcommand{\fpkp}{\CF_{p}(k, p)}

\newcommand{\BC}{{\mathbb C}}
\newcommand{\BF}{{\mathbb F}}

\newcommand{\BR}{{\mathbb R}}

\newcommand{\fp}{\BF_{p}}
\newcommand{\fq}{\BF_{q}}

\newcommand{\fqn}{\BF_{q}^{n}}
\newcommand{\fqx}{\BF_{q}[X]}

\newcommand{\dH}{d_{{\rm H}}}
\newcommand{\dc}{d_{{\rm ch}}}
\newcommand{\rhoh}{\rho_{{\rm H}}}
\newcommand{\rhos}{\rho_{{\rm ch}}}
\newcommand{\rhosbar}{\overline{\rho}_{{\rm ch}}}
\newcommand{\rhohbar}{\overline{\rho}_{{\rm H}}}

\newcommand{\abs}[1]
{{\raisebox{-0.25\depth}{$\biggl\lvert$}}{#1}\raisebox{-0.25\depth}{$\biggr\rvert$}}
\newcommand{\norm}[1]
{\left\|{#1}\right\|}

\newcommand{\bigabs}[1]
{{\raisebox{-0.25\depth}{$\biggl\lvert$}}{#1}\raisebox{-0.25\depth}{$\biggr\rvert$}}

\newcommand{\herm}[1]
{{#1}^{\dagger}}

\DeclareMathOperator{\betaf}{B}
\DeclareMathOperator{\Tr}{Tr}
\DeclareMathOperator{\tr}{tr}
\DeclareMathOperator{\rank}{rank}
\DeclareMathOperator{\Image}{Im}

\DeclareMathOperator{\CCP}{CP}
\DeclareMathOperator{\CRS}{CRS}
\DeclareMathOperator{\RS}{RS}
\DeclareMathOperator{\GRS}{GRS}

\DeclareMathOperator{\Vol}{Vol}
\DeclareMathOperator*{\expec}{\mathbb{E}}
\DeclareMathOperator{\GS}{GS}
\DeclareMathOperator{\GRScover}{GRS-cover}

\DeclareMathOperator{\CRScover}{CRS-cover}
\DeclareMathOperator{\GRSdecode}{GRS-decode}

\DeclareMathOperator{\polylog}{polylog}

\newtheorem{theorem}{Theorem}
\newtheorem{lemma}{Lemma}

\newtheorem{corollary}{Corollary}
\newtheorem{conjecture}{Conjecture}

\theoremstyle{definition}
\newtheorem{definition}{Definition}

\theoremstyle{remark}

\newcommand{\hide}[1]
{{\iffalse #1 \fi}}

\IEEEoverridecommandlockouts
\title{Covering in Hamming and Grassmann Spaces: New Bounds and Reed–Solomon-Based Constructions}
\author{
    Samin Riasat, 
    \IEEEmembership{Student Member, IEEE}, and 
    Hessam Mahdavifar, 
    \IEEEmembership{Member, IEEE}
    \thanks{
        Samin Riasat and Hessam Mahdavifar are with the Department of Electrical and Computer Engineering, Northeastern University, Boston, MA 02115, USA
        (email: \{
        \href{mailto:riasat.s@northeastern.edu}{riasat.s}, \href{mailto:h.mahdavifar@northeastern.edu}{h.mahdavifar}\}@northeastern.edu). 
    }
    \thanks{
        This paper was presented in part at the IEEE International Symposium on Information Theory in June 2025 \cite{riasat2025covering}. This work was supported by NSF under Grant CCF-2415440 and the Center for Ubiquitous Connectivity (CUbiC) under the JUMP 2.0 program.
    }
}
\date{}

\begin{document}

\IEEEtitleabstractindextext{%
    \begin{abstract}
        We study covering problems in Hamming and Grassmann spaces through a unified coding-theoretic and information-theoretic framework. Viewing covering as a form of quantization in general metric spaces, we introduce the notion of the average covering radius as a natural measure of average distortion, complementing the classical worst-case covering radius. By leveraging tools from one-shot rate-distortion theory, we derive explicit non-asymptotic random-coding bounds on the average covering radius in both spaces, which serve as fundamental performance benchmarks.
        
        On the construction side, we develop efficient puncturing-based covering algorithms for generalized Reed–Solomon (GRS) codes in the Hamming space and extend them to a new family of subspace codes, termed character-Reed–Solomon (CRS) codes, for Grassmannian quantization under the chordal distance. Our results reveal that, despite poor worst-case covering guarantees, these structured codes exhibit strong average covering performance. In particular, numerical results in the Hamming space demonstrate that RS-based constructions often outperform random codebooks in terms of average covering radius. In the one-dimensional Grassmann space, we numerically show that CRS codes over prime fields asymptotically achieve average covering radii within a constant factor of the random-coding bound in the high-rate regime. Together, these results provide new insights into the role of algebraic structure in covering problems and high-dimensional quantization.
    \end{abstract}
    
}
\maketitle
\IEEEdisplaynontitleabstractindextext
\begingroup\renewcommand\thefootnote{\textsection}
\endgroup


\section{Introduction}
\label{sec:intro}

Covering problems lie at the heart of several questions in coding theory, information theory, and signal processing, where one seeks structured families of objects that approximate the ambient space with minimal distortion. 
They are mathematically appealing in their own right, but have also found a wide range of technical applications in data compression, signal processing, clustering, sampling, and robust system design, especially in high-dimensional spaces, to name a few~\cite{torquato10, toth22}. 
In discrete settings, such as the Hamming space, the covering radius of a code captures its ability to represent arbitrary vectors, while in geometric settings like the Grassmann space, analogous notions arise through metrics such as the chordal distance on subspaces. Despite their shared conceptual foundation, these two domains have typically been treated with separate analytic tools and constructions. This work develops a unified perspective on covering in Hamming and Grassmann spaces, focusing on both the algorithmic task of finding nearby codewords and the quantitative task of understanding the distortion inherent to different code families.

The notion of subspace codes and their application in non-coherent communication was first developed in the seminal work of Koetter and Kschischang~\cite{KK} in the context of randomized network coding. Specifically, they developed a coding theory framework for subspace coding over finite fields. Inspired by this, a parallel framework in the field of real/complex numbers was later developed by Soleymani and Mahdavifar~\cite{Hessam22}. They introduced a new algebraic construction for one-dimensional complex subspace codes called \emph{character-polynomial (CP) codes} and later extended it to higher dimensions~\cite{Hessam21}. 

In the Hamming space, we study the covering problem through the lens of Reed--Solomon (RS) codes \cite{reed1960polynomial}. This well-known family of codes is not only theoretically rich but also widely used in practice, ranging from hard-disk drives to distributed storage and computing. 
Although various properties of RS codes have been studied extensively in the literature, they continue to inspire the development of new codes to this day. 
This includes the aforementioned CP codes \cite{Hessam22}, which are essentially the concatenation of a subcode of \emph{generalized} Reed--Solomon (GRS) code, where a certain subset of the coefficients of the message polynomial is set to zeros, with a character function mapping the codeword coordinates to the complex unit circle. In a recent work, Gooty et al.~\cite{gooty2025precodingdesignlimitedfeedbackmiso} used CP codes to construct new precoding designs for multiple-input multiple-output (MIMO) systems with limited feedback, an important problem in wireless communications. 


\subsection{Problem Motivation}

Covering can be viewed as a form of quantization in a general metric space, which is central to problems involving the approximation or representation of continuous (or large discrete) spaces using structured, finite subsets. In a canonical covering problem, one seeks to cover a metric space with the minimum number of spheres of a prescribed radius, referred to as the covering radius. This guarantees that any point in the space can be approximated by mapping it to the center of a nearby sphere within the covering radius.

A natural approach to the covering problem is to leverage well-established code families in the underlying metric space that possess strong minimum-distance properties, while shifting attention to the algorithmic challenge of quantization. Specifically, the goal is to design efficient procedures that map an arbitrary element of the space to a nearby codeword, without requiring optimal (nearest-neighbor) decoding. In the Hamming space, GRS codes are known to meet the redundancy bound~\cite[Corollary~11.1.3]{Huffman03}, implying that their worst-case covering radius is large and therefore unattractive from a classical covering perspective. To the best of our knowledge, this remains the only general covering result available for RS codes.

Despite this unfavorable worst-case guarantee, we show in this paper that the average covering radius of GRS codes is significantly smaller than the redundancy bound. This observation highlights a fundamental gap between worst-case and typical covering behavior and provides strong motivation for developing \textit{efficient covering algorithms} for RS-based constructions. Moreover, since subspace codes can be obtained by mapping RS codes to the complex domain, these insights naturally extend to covering problems in Grassmann spaces, with direct relevance to applications such as limited-feedback wireless communication systems.

\subsection{Related Work}

The theory of covering codes has a rich history in classical information theory. 
Early work in this area includes characterizing upper and lower bounds on the covering radius and its computational complexity~\cite{Cohen85, McLoughlin84}, bounds on the minimal size of a code of a given covering radius~\cite{Cohen86, vanLint88}, the complexity of bounding the covering radius of a binary code~\cite{Cohen97}, the relationship between the covering radius of a code and its subcodes~\cite{Brualdi98}, the covering radius of extended Reed--Solomon codes~\cite{Dur94, Dur91}, among others. We refer the reader to the extensive list of references following \cite[Chapter~15]{Huffman03}. More recently, there has been a renewed interest in determining the covering radius of various classes of codes, including Melas codes \cite{shi2022covering}, Zetterberg type codes \cite{shi2023covering,shi2025determining}, Euclidean codes \cite{shi2025covering}, and extended algebraic-geometry codes \cite{zhu2025dual}, as well as bounds on the expected Hamming distortion of linear codes over $\mathbb{F}_{2}$ and lattices in $\mathbb{R}^{n}$~\cite{ordentlich2026voronoisphericalcdflattices}. However, none of these recent works studied the problem of designing \textit{efficient covering algorithms} and only focused on determining the covering radius. Related works in the subspace domain include asymptotic bounds and constructions~\cite{barg2002bounds, blackburn2011asymptoticbehaviorgrassmanniancodes}, covering bounds in Grassmann spaces over finite fields \cite{qian2022coveringgrassmanniancodesbounds}, efficient decoding algorithms for CP codes~\cite{riasat2024decodinganalogsubspacecodes} and, more recently, sensing subspace codes~\cite{mahdavifar2024subspace,gooty2025efficientdecoderssensingsubspace}.

Despite the extensive early literature and the recent renewed interest in determining the covering radius of specific code families, the design and analysis of \textit{efficient covering algorithms} has remained largely unexplored. Moreover, existing works predominantly focus on worst-case covering guarantees, leaving the typical or average-case behavior of covering schemes less understood. We hope that the present work helps renew interest in the covering problem by shifting attention toward algorithmic and average-case perspectives, particularly in light of its growing relevance to modern applications such as quantization, lossy compression, limited-feedback communication systems, and high-dimensional signal representation.

\subsection{Our Contribution}

In this paper, we develop a unified algorithmic and information-theoretic framework for covering in both Hamming and Grassmann spaces. On the algorithmic side, we propose an efficient puncturing-based covering algorithm for GRS codes. The algorithm operates by successively puncturing the code and invoking an off-the-shelf decoder on the resulting shortened codes, until a nearby codeword is identified. We show that this procedure efficiently, i.e., with polynomial-time complexity, finds, for any input vector, a codeword within the covering radius of the input vector. We analyze the average-case performance of the proposed algorithm by characterizing the average number of punctures required for success and by estimating the resulting average covering radius, defined as the expected distance between the input vector and the selected codeword. In the Hamming space, we derive theoretical bounds on the average number of punctures and partially resolve a conjecture posed in the preliminary version of this work~\cite[Conjecture~1]{riasat2025covering}. These results demonstrate that, despite their poor worst-case covering guarantees, GRS codes exhibit favorable average covering behavior.

To extend these ideas to subspace covering in the complex domain, we introduce character-Reed–Solomon (CRS) codes, which generalize CP codes by removing structural constraints tailored to packing. We show how the proposed puncturing-based algorithm can be adapted to CRS codes, yielding an efficient covering method for one-dimensional Grassmannian spaces over complex numbers under the chordal distance.

From an information-theoretic perspective, we formalize the notion of average covering radius as a natural measure of average distortion when a code is used for quantization. Building on one-shot rate-distortion results of Elkayam and Feder~\cite{elkayam20}, we derive explicit non-asymptotic random-coding bounds on the average covering radius in both Hamming and Grassmann spaces. A numerical analysis of CRS codes over prime fields suggests that, at extreme high rates, these algebraic constructions asymptotically achieve average covering radii within a constant factor of the random-coding bound in the one-dimensional Grassmann space. Moreover, numerical results in the Hamming space demonstrate that RS-based codes frequently outperform random codebooks with respect to the average covering radius.

The rest of the paper is organized as follows. 
In \autoref{sec:preliminaries} we provide the necessary background on GRS and subspace codes. 
In \autoref{sec:lossy} we define the average covering radius and explore its connection to the average distortion measure in lossy source coding.
In \autoref{sec:covering-grs} we present the covering algorithm for GRS codes along with a theoretical analysis of \hide{the average number of punctures and the fraction of coverage of the ambient space,}its performance as well as numerical results. 
In \autoref{sec:CRS-code} we introduce CRS subspace codes and present a covering algorithm for these codes.
Finally, we conclude the paper in \autoref{sec:conclusion} with remarks on possible future directions. 

\section{Preliminaries}
\label{sec:preliminaries}

Throughout this paper, $q$ denotes a positive integer power of a prime number $p$. 
Given a code $\CC \subseteq \fqn$ and a codeword $c \in \CC$, we denote the Hamming sphere of radius $\tau$ centered at $c$ by $B(c, \tau)$ and its volume by $\Vol_{q}(\tau, n) := |B(c, \tau)|$.

\subsection{Reed--Solomon Code}

Fix $k \le n \le q$. Consider the \emph{message space}:
\begin{align*}
    \CF(k, q) 
    &:= \{f \in \fqx: \deg(f) < k\}, 
\end{align*}
which consists of all polynomials of degree less than $k$ over $\fq$. 
The elements of $\CF(k, q)$ are called \emph{message polynomials}, whose coefficients represent message symbols. 
The coefficient of $X^{j}$ in $f \in \CF(k, q)$ is denoted by $f_{j}$. 
    
Given distinct $\alpha_{1}, \dots, \alpha_{n} \in \fq$ (called \emph{evaluation points}), the encoding of $f \in \CF(k, q)$ 
in the \emph{Reed--Solomon code} $\RS := \RS_{n}(\CF(k, q))$ of length $n$ and dimension $k$ over $\fq$ is defined as follows:
\begin{align}
    \label{eq:rs}
    \RS(f) 
    &:= (f(\alpha_{1}), \dots, f(\alpha_{n})). 
\end{align}
In addition, given arbitrary 
$v_{1}, \dots, v_{n} \in \fq^{\times}$ (called \emph{column multipliers}), the encoding of $f \in \CF(k, q)$ in the \emph{generalized Reed--Solomon code} $\GRS := \GRS_{n}(\CF(k, q))$ of length $n$ and dimension $k$ over $\fq$ is defined as follows:
\begin{align}
    \label{eq:grs}
    \GRS(f) 
    &:= (v_{1} f(\alpha_{1}), \dots, v_{n} f(\alpha_{n})).
\end{align}

$\RS$ and $\GRS$ above are well known to be \emph{maximum distance separable (MDS)}, i.e., they are $[n, k, d]_{q}$ codes 
with $d := n - k + 1 \ge 1$. 
Note also that $\GRS = \RS$ when $v_{i} = 1$ for all $i$. 

\subsection{Chordal Distance and Subspace Code}

For an ambient vector space $W$, we use $\CP(W)$ and $\CP_{m}(W)$ to denote the set of all subspaces of $W$ and the set of all $m$-dimensional subspaces of $W$, respectively. In particular, $\CP_{m}(\BC^{n})$ is referred to as a \emph{Grassmann space} and is denoted by $G_{m, n}(\BC)$. \hide{The elements of $G_{m, n}(\BC)$ are called \emph{$m$-planes}. }Any \hide{$m$-plane }$U \in G_{m, n}(\BC)$ is equipped with the natural inner product $\langle u, v\rangle := \herm{u} v$ for $u, v \in U$. 

\begin{definition}[{\cite[\S~3]{Conway96,Hessam22}}]
    \label{def:chordal}
    Given $U, V \in G_{m, n}(\BC)$\hide{ be $m$-planes}, let $u_i \in U$ and $v_i \in V$ for $i \in \{1, \dots, m\}$ such that $|\langle u_i, v_i\rangle|$ is maximal, subject to the condition that they form orthonormal bases for $U$ and $V$, respectively. The $i$-th \emph{principal angle} $\theta_i$ between $U$ and $V$ is defined as $\theta_i := \arccos|\langle u_{i}, v_{i}\rangle|$. Then, the \emph{chordal distance} between $U$ and $V$ is defined as follows:
    \begin{align}
        \label{eq:dc}
        \dc(U, V) 
        := \sqrt{\sum_{i = 1}^{m} \sin^{2} \theta_{i}}. 
    \end{align}
    We define a \emph{subspace code} (also called \emph{Grassmann code}) to be a collection $\CC \subseteq G_{m, n}(\BC)$ of subspaces with respect to the chordal distance. 
\end{definition}

Soleymani and Mahdavifar~\cite{Hessam22} introduced a new class of subspace codes called \emph{character-polynomial (CP) codes} defined as follows. First, we recall the definition of a \emph{character}. 

\begin{definition}[Character]
    \label{def:character}
    A homomorphism $\chi$ from the additive group of $\fq$ to the complex unit circle $U$ is called an \emph{(additive) character} of $\fq$. 
    For $\beta \in \fq$, the character $\chi_{\beta}$ of $\fq$ is defined as follows:
    \begin{align}
        \label{eq:character}
        \chi_{\beta}: \fq \to U, \quad \alpha \mapsto \exp\left(\frac{2 \pi i \Tr(\beta \alpha)}{p}\right),
    \end{align} 
    where $\Tr: \fq \to \fp$ is the \emph{absolute trace function}. 
    It is well-known that every character $\chi$ of $\fq$ is of the form $\chi = \chi_{\beta}$ for some $\beta \in \fq$.
    The character $\chi_{0}$ \hide{when $\beta = 0$ }is called the \emph{trivial character}. 
\end{definition}

\begin{definition}[CP Code {\cite[Definition~6]{Hessam22}}]
    \label{def:cp}
    Fix $k \le n < q$, a non-trivial character $\chi_{\beta}$ of $\fq$, and units $\alpha_{1}, \dots, \alpha_{n} \in \fq^{\times} := \fq \setminus \{0\}$. 
    Define $\fpkq$ to be the set of all polynomials \hide{$f(X) = \sum_{j} f_{j} X^{j}$}$f \in \CF(k, q)$ with $f_{j p} = 0$ for all integers $j \ge 0$\hide{, and put $\fpkq' := \{f(X) / X: f \in \fpkq\}$}. 
    Then, the encoding of $f \in \fpkq$ in $\CCP := \CCP_{n, \beta}(\fpkq) \subseteq G_{1, n}(\BC)$ is defined as follows:
    \begin{align}
        \label{eq:cp}
        \CCP(f) 
        &:= (\chi_{\beta}(f(\alpha_{1})), \dots, \chi_{\beta}(f(\alpha_{n}))), 
    \end{align}
    where $\CCP(f)$ is identified with the one-dimensional subspace $\langle \CCP(f) \rangle$. 
\end{definition}


Note that the structure of the message space $\fpkq$ is crucial for orthogonality and distance properties that make CP codes useful for packing subspaces in $G_{1, n}(\BC)$~\cite{Hessam22, Hessam21}. 
Recall that in the \emph{sphere packing} problem, the goal is to fit in the space pairwise disjoint spheres centered at the codewords. This is fundamentally related to decoding, where, given a point in the space, the objective is to uniquely determine its nearest codeword. In particular, given a packing of the space where the spheres have radius $\tau < d / 2$, where $d$ is the minimum distance of the code, one can uniquely decode any message from $\tau$ errors, and the goal is to maximize the decoding radius $\tau$. The dual problem to sphere packing is \emph{sphere covering}, 
which is the subject of our discussion in this work. 
In the sphere covering problem, the goal is to cover the space using minimally overlapping spheres. 
This is fundamentally related to quantization, which is the process of approximating any given point in the space by its closest codeword.

\section{Covering Radius as a Measure of Distortion\hide{ and Comparison with Random Coding Bound}}
\label{sec:lossy}

\subsection{Average Covering Radius}
\label{sec:covering-radius}

The \emph{(Hamming) covering radius} of a block code $\CC \subseteq \fq^{n}$ is defined as follows~\cite{Huffman03}:
\begin{align}
    \label{eq:rho}
    \rhoh(\CC) 
    &:= \max_{y \in \fqn} 
    \min_{c \in \CC} 
    \dH(y, c), 
\end{align}
where $\dH$ denotes Hamming distance. 
For instance, it is well known that $\rhoh(\GRS) = d - 1$. 
This means, in particular, that any $y \in \fq^{n}$ is at distance at most $d - 1$ from a $\GRS$ codeword. 

Given a vector $y \in \fq^{n}$ and a code $\CC \subseteq \fq^{n}$, the \emph{covering problem} asks to find a codeword $c \in \CC$ such that $\dH(y, c) \le \rhoh(\CC)$. 
An algorithm that solves the covering problem for $\CC$ is called a \emph{covering algorithm} for $\CC$. 

Generally speaking, the covering radius of a code represents the maximum error when the code is used for the quantization of its ambient space\hide{ as described in \autoref{sec:intro}}. This is because the covering radius characterizes the worst-case scenario of the quantization process. In practice, however, given a certain distribution over the space, the average quantization error becomes more relevant. This motivates us to modify \eqref{eq:rho} to the following definition. 


\begin{definition}[Average Covering Radius]
    \label{def:mean-covering-radius}
    For a block code $\CC \subseteq \fqn$ and a distribution $\CD$ on $\fqn$, the \emph{average (Hamming) covering radius} of $\CC$ over $\CD$ is defined as follows: 
    \begin{align}
        \label{eq:rhobar}
        \rhohbar(\CC) 
        &:= \expec_{X \sim \CD} 
        \left[\min_{c \in \CC} 
        \dH(X, c)\right]. 
    \end{align}
    Similarly, for a subspace code $\CC \subseteq G_{m, n}(\BC)$ and a distribution $\CD$ on $G_{m, n}(\BC)$, the \emph{average (chordal) covering radius} of $\CC$ over $\CD$ is defined as follows: 
    \begin{align}
        \label{eq:rhosbar}
        \rhosbar(\CC) 
        &:= \expec_{U \sim \CD} 
        \left[\min_{V \in \CC} 
        \dc(U, V)\right].
    \end{align}
\end{definition}

Note that \autoref{def:mean-covering-radius} can be generalized naturally to codes over other metric spaces, e.g., rank-metric codes. Henceforth, \emph{covering radius} will be understood to mean \emph{Hamming covering radius} for the Hamming space, and \emph{chordal covering radius} for the Grassmann space. 

\subsection{Average Distortion of Random Codes\hide{in Lossy Source Coding}}

In the single-shot approach to the channel coding (or lossy compression) problem, the goal is to control a random object by providing achievable bounds without any assumption on the object itself. In particular, let $\CX$ be a set of \emph{input symbols} and $\CY$ a set of \emph{reproduction symbols}. Given a code $\CC \subseteq \CY$ and distribution $\CD_{\CX}$ on $\CX$, Elkayam and Feder~\cite{elkayam20} defined the \emph{average distortion} associated with $\CC$ as follows: 
\begin{align}
    \label{eq:average-distortion-def} 
    D(\CC) 
    &:= \expec_{X \sim \CD_{\CX}} 
    \left[\min_{c \in \CC} d(X, c)\right], 
\end{align}
where $d: \CX \times \CY \to \BR^{+}$ is the \emph{distortion function}. 
They gave the following explicit formula for the average distortion of a \emph{random code} $\CC_{\CY} := \{Y_{j} \sim \CD_{\CY}: 1 \le j \le M\}$ consisting of $M$ codewords $Y_{j}$ drawn independently from a distribution $\CD_{\CY}$ on $\CP(\CY)$~\cite[Theorem~2]{elkayam20}: 
\begin{align}
    \label{eq:elkayam-formula}
    &\hspace{-2mm}\expec_{\CC_{\CY}} [D(\CC_{\CY})] 
    = \int_{0}^{1} 
    \Tilde{D}(w, \CD_{\CY}) G_{M}'(w) \, d w, 
\end{align} 
where 
\begin{align}
    \label{eq:D-tilde} 
    &\hspace{-5mm} \Tilde{D}(w, \CD_{\CY}) := \\
    &\hspace{.1\linewidth} \frac{1}{w} \cdot \expec_{(X, Y) \sim \CD_{\CX} \times \CD_{\CY}}[d(X, Y) \cdot \mathbf{1}_{\{p_{X, Y, U} \le w\}}], \notag \\ 
    \label{eq:pxyu} 
    &\hspace{-5mm} p_{x, y, U} 
    := \Pr_{Y \sim \CD_{\CY}}[d(x, Y) < d(x, y)] \\ 
    &\hspace{.27\linewidth} + U \cdot \Pr_{Y \sim \CD_{\CY}}[d(x, Y) = d(x, y)], \notag \\ 
    \label{eq:g} 
    &\hspace{-5mm} G_{M}(w) 
    := - (1 - w)^{M - 1} 
    \cdot [1 + (M - 1) w], 
\end{align}
and $U$ is a uniform random variable in $[0, 1]$. 
Henceforth, to simplify notation, we will usually omit the distribution under $\Pr$ or $\expec$ when it is clear from the random variable argument(s). 

The distance function of a code serves as a natural measure of distortion when the code is used for the quantization of its ambient space. Under this identification, the average distortion $D(\CC)$, defined in \eqref{eq:average-distortion-def}, of a code $\CC$ coincides with the average covering radius of $\CC$ defined in \autoref{def:mean-covering-radius}, i.e., $\rhohbar(\CC)$ in \eqref{eq:rhobar} for the Hamming space ($\CX = \CY = \fqn$, $d = \dH$), and $\rhosbar(\CC)$ in \eqref{eq:rhosbar} for the Grassmann space ($\CX = \CY = G_{1, n}(\BC)$, $d = \dc$). 
In particular, we can use \eqref{eq:elkayam-formula} to estimate the average covering radius of random codes in both spaces as shown by the next two theorems. 

\begin{theorem}
    \label{thm:random-hamming}
    Let $\CD_{\CY}$ denote the uniform distribution on the Hamming space $\CY = \fqn$. 
    Then, the expected average covering radius \hide{over $\CD_{\CX}$ }of a random code $\CC_{\CY}$ consisting of $M$ codewords drawn independently from $\CD_{\CY}$ 
    is given by:
    \begin{align}
        \label{eq:random-hamming}
        &\expec_{\CC_{\CY}} [\rhohbar(\CC_{\CY})] = \\ 
        &\quad M (M - 1) \sum_{j = 1}^{n} \bigg[j \left(A_{M, j} - A_{M - 1, j}\right) 
        + A_{M - 1, j} \sum_{t = 0}^{j - 1} V_{t}\bigg], 
        \notag 
    \end{align} 
    where 
    \begin{align}
        \label{eq:Vjdef}
        &V_{j} 
        := \frac{\Vol_{q}(j, n)}{q^{n}}, \\ 
        \label{eq:Adef} 
        &A_{i, j} 
        := \frac{(1 - V_{j})^{i} - (1 - V_{j - 1})^{i}}{i}. 
    \end{align}
\end{theorem}

\begin{proof}
    Note that $\Pr[\dH(x, Y) \le j] = V_{j}$\hide{ for $Y \sim \CD_{\CY}$} for any $x \in \fqn$ by \eqref{eq:Vjdef}. 
    Hence, 
    by \eqref{eq:pxyu}, 
    \begin{align}
        p_{x, y, U} 
        &= V_{t - 1} + U (V_{t} - V_{t - 1}), 
        \label{eq:pxyu-hamming}
    \end{align}
    where $t = \dH(x, y)$ and $U$ is uniform in $[0, 1]$. 
    
    For $(X, Y) \sim \CD_{\CX} \times \CD_{\CY}$, 
    where $\CD_{\CX}$ is any distribution on $\CX = \CY$, 
    let $E_{t}$ denote the event $\dH(X, Y) = t$. 
    Then, 
    by \eqref{eq:D-tilde}, 
    \begin{align}
        \notag 
        \Tilde{D}(w, \CD_{\CY}) 
        &= \frac{1}{w} \sum_{t = 1}^{n} t \Pr[\dH(X, Y) \cdot {\bf 1}_{\{p_{X, Y, U} \le w\}} = t] \\ 
        \notag 
        &= \frac{1}{w} \sum_{t = 1}^{n} t \Pr[E_{t} \hbox{  and  } p_{X, Y, U} \le w] \\ 
        \label{eq:D-tilde-thm-1}
        &= \frac{1}{w} \sum_{t = 1}^{n} t \Pr[E_{t}] \cdot \Pr[p_{X, Y, U} \le w | E_{t}].  
    \end{align}
    Now, 
    by \eqref{eq:Vjdef}, 
    \begin{align}
        \label{eq:pret}
        \hspace{-2mm}
        \Pr[E_{t}] 
        &= \expec[\Pr[E_{t} | X]] 
        = \expec [V_{t} - V_{t - 1}]
        = V_{t} - V_{t - 1}
    \end{align} 
    and, 
    by \eqref{eq:pxyu-hamming}, 
    \begin{align}
        \notag 
        \hspace{-2mm}
        \Pr[p_{X, Y, U} \le w | E_{t}] 
        &= \Pr\left[U \le \frac{
        w - V_{t - 1}
        }{
        V_{t} - V_{t - 1}}\bigg| E_{t}\right] \\ 
        \label{eq:pret-cond}
        &= \begin{cases}
            0 &\hbox{if } w \le V_{t - 1}
            \\ 
            \dfrac{
            w - V_{t - 1}
            }{
            V_{t} - V_{t - 1}}
            & \hbox{if } V_{t - 1} < w \le V_{t}
            \\ 
            1 & \hbox{if } w > V_{t}.
        \end{cases}
    \end{align}
    Substituting \eqref{eq:pret} and \eqref{eq:pret-cond} into \eqref{eq:D-tilde-thm-1} yields: 
    \begin{align}
        \label{eq:D-tilde-hamming-uniform}
        \Tilde{D}(w, \CD_{\CY}) 
        = t_{0} \left(1 - \frac{V_{t_{0} - 1}}{w}\right) 
        + \sum_{t = 1}^{t_{0} - 1} \frac{t}{w} (V_{t} - V_{t - 1})
        ,
    \end{align}
    where $t_{0} := t_{0}(w) \in \{1, \dots, n\}$ such that 
    $V_{t_{0} - 1} < w \le V_{t_{0}}$. 
    Substituting \eqref{eq:D-tilde-hamming-uniform} and $G'(w) = M (M - 1) w (1 - w)^{M - 2} d w$ into \eqref{eq:elkayam-formula} gives: 
    \begin{align*}
        &\expec_{\CC_{\CY}} [\rhohbar(\CC_{\CY})] 
        = \int_{0}^{1} 
        \Tilde{D}(w, \CD_{\CY}) G'(w) d w \\ 
        &= M (M - 1) \sum_{j = 1}^{n} \bigg[j \int_{V_{j - 1}}^{V_{j}} 
        \bigg(1 - \frac{V_{j - 1}}{w}\bigg) w (1 - w)^{M - 2} d w  \\ 
        &\hspace{2.40cm} + \sum_{t = 1}^{j - 1} t (V_{t} - V_{t - 1}) \int_{V_{j - 1}}^{V_{j}} (1 - w)^{M - 2} d w\bigg] \\ 
        &= M (M - 1) \sum_{j = 1}^{n} \bigg[j \underbrace{\int_{V_{j - 1}}^{V_{j}} w (1 - w)^{M - 2} d w}_{A_{M, j} - A_{M - 1, j}} \\ 
        &\hspace{5mm}- \bigg(j V_{j - 1} - \sum_{t = 1}^{j - 1} t (V_{t} - V_{t - 1})\bigg) \underbrace{\int_{V_{j - 1}}^{V_{j}} (1 - w)^{M - 2} d w}_{A_{M - 1, j}}\bigg] \\ 
        &= M (M - 1) \sum_{j = 1}^{n} \bigg[j \left(A_{M, j} - A_{M - 1, j}\right) + A_{M - 1, j} \sum_{t = 0}^{j - 1} V_{t}\bigg], 
    \end{align*}
    where we have used \eqref{eq:g} and \eqref{eq:Adef}. 
\end{proof}

To obtain the analog of \autoref{thm:random-hamming} for the Grassmann space $\CY = G_{1, n}(\BC)$, 
we first introduce the following notation. 
For $V = \langle v \rangle \in \CY$ with $\norm{v} = 1$, the orthogonal projection operator onto $V$ is denoted by $P_{V} := v v^{\dagger}$. 
Some key observations are: 
\begin{itemize}
    \item For any $x, y \in G_{1, n}(\BC)$, 
    \begin{align}
        \label{eq:tr-chordal}
        \dc(x, y) 
        &= \sqrt{1 - \tr(P_{x} P_{y})}.  
    \end{align}
    \item $p_{x, y, U} = \Pr[\dc(x, Y) < \dc(x, y)]$ for any continuous distribution $\CD_{\CY}$ by \eqref{eq:pxyu}, which is independent of $U$, and 
    \item $\tr(P_{x} P_{Y}) \hide{= \tr(P_{\langle e_{1} \rangle} Q P_{Y} Q^{\dagger}) }= \tr(P_{Q x} P_{Q Y}) = |\hat{Y}_{1}|^{2} \sim {\rm Beta}(1, n - 1)$ by unitary invariance, 
    for any fixed $x := \langle \hat{x} \rangle \in \CY$, $Y$ Haar-uniform on $\CY$, and unitary $Q$ with $Q \hat{x} = \hat{e}_{1}$, 
    so that the distributions of $Y$ and $Q Y := \langle \hat{Y} \rangle$ are equal~\cite{kus88}. 
\end{itemize}
We then have the following result. 

\begin{theorem}
    \label{thm:random-chordal} 
    Let $\CD_{\CY}$ denote the uniform distribution on the Grassmann space $\CY = G_{1, n}(\BC)$. 
    Then, the expected average covering radius \hide{over $\CD_{\CX}$ }of a random code $\CC_{\CY}$ consisting of $M$ codewords drawn independently from $\CD_{\CY}$ 
    is given by:
    \begin{multline}
        \label{eq:random-chordal}
        \expec_{\CC_{\CY}} [\rhosbar(\CC_{\CY})] = \\ 
        \frac{2 (n - 1)}{2 n - 1} M (M - 1) \betaf\left(1 + \frac{2 n - 1}{2 (n - 1)}, M - 1\right), 
    \end{multline} 
    where $\betaf(\cdot, \cdot)$ is the beta function. 
\end{theorem}

\begin{proof}
    For $(X, Y) \sim \CD_{\CX} \times \CD_{\CY}$, where $\CD_{\CX}$ is any distribution on $\CX = \CY$, let $Z := \tr(P_{X} P_{Y})\hide{|\langle \hat{x}, \hat{y} \rangle|^{2}$ for unit-norm representatives  $\hat{x} \in X$ and $\hat{y} \in Y}$.  
    Then, for any $x \in \CX$, we have $Z \mid X = x \sim {\rm Beta}(1, n - 1)$, 
    so that $Z \sim {\rm Beta}(1, n - 1)$. Therefore, 
    by \eqref{eq:tr-chordal},
    \begin{align}
        p_{x, y, U} 
        &= \Pr[\dc(X, Y) < t] \notag \\ 
        &= \Pr[Z > 1 - t^{2}] \notag \\ 
        \label{eq:pxyu-chordal}
        &= t^{2 (n - 1)}, 
    \end{align} 
    where $t = \dc(x, y)$. 
    Letting $T := \dc(X, Y)\hide{ = \sqrt{1 - Z}}$ therefore gives: 
    \begin{align}
        \Tilde{D}(w, \CD_{\CY}) 
        &= \frac{1}{w} \expec[\dc(X, Y) \cdot {\bf 1}_{\{p_{X, Y, U} \le w\}}] \notag \\ 
        &= \frac{1}{w} \expec[T \cdot {\bf 1}_{\{T^{2 (n - 1)} \le w\}}] \notag \\ 
        &= \frac{1}{w} \int_{0}^{w^{\frac{1}{2 (n - 1)}}} t f_{T}(t) d t \notag \\ 
        &= \frac{2 (n - 1)}{w} \int_{0}^{w^{\frac{1}{2 (n - 1)}}} t^{2 n - 2} d t \notag \\ 
        &= \frac{2 (n - 1)}{2 n - 1} w^{\frac{1}{2 (n - 1)}}, 
        \label{eq:D-tilde-chordal}
    \end{align}
    which follows using \eqref{eq:pxyu-chordal}. 
    Thus, by \eqref{eq:elkayam-formula}, 
    \begin{align*} 
        &\expec_{\CC_{\CY}}[\rhosbar(\CC_{\CY})] 
        = \int_{0}^{1} \Tilde{D}(w, \CD_{\CY}) G'(w) d w \\ 
        &= \frac{2 (n - 1)}{2 n - 1} M (M - 1) \int_{0}^{1} w^{1 + \frac{1}{2 (n - 1)}} (1 - w)^{M - 2} d w \\ 
        &= \frac{2 (n - 1)}{2 n - 1} M (M - 1) \betaf\left(1 + \frac{2 n - 1}{2 (n - 1)}, M - 1\right), 
    \end{align*}
    which follows using \eqref{eq:g} and \eqref{eq:D-tilde-chordal}, as desired. 
\end{proof}

It is important to note that the conclusions of \autoref{thm:random-hamming} and \autoref{thm:random-chordal} do not depend on the source distribution $\CD_{\CX}$ over which the average covering radius of the random code $\CC_{\CY}$ is defined. 

\section{Covering Using GRS Codes}
\label{sec:covering-grs}

Decoding algorithms have been the subject of study for many decades, and several decoding algorithms for GRS codes are known~\cite{Huffman03, Guruswami06}. 
In contrast, to the best of our knowledge, no explicit covering algorithms are known for GRS codes.  
Note that covering involves solving a more intricate minimax problem, which is computationally harder due to the need to analyze the worst-case scenario over the entire space. 
Although decoding algorithms may be used for covering, their performance is bound to be subpar as spheres of radius $\tau < d / 2$ do not cover the space well. 
Therefore, as the first algorithm of its kind to the best of our knowledge, we propose \autoref{alg:GRS-cover} and analyze its performance in this section. 

\subsection{Algorithm}

Given an $[n, k, d]_{q}$ GRS code $\CC$ and an input vector $y$ in its ambient space, the algorithm first tries to decode $y$. If the decoder is successful, then we have found the closest codeword to $y$. If not, it repeatedly (up to $d - 1$ times) punctures the code at the last coordinate and tries to decode the corresponding substring of $y$ until the decoder is successful,
when it finally returns the re-encoding of the decoder output in the original code. 
Note that in an actual implementation of the algorithm, it may be desirable to have the message polynomial returned directly instead of its re-encoding, as we will see in \autoref{sec:covering-cp}. 

\begin{algorithm}[!htbp]
    \caption{$\GRScover(\CC, y)$} 
    \label{alg:GRS-cover}
    \begin{algorithmic}[1]
        \Require{$[n, k, d]_{q}$ $\GRS$ code $\CC$, arbitrary vector $y \in \fq^{n}$}
        \Ensure{A codeword $c \in \CC$ with $\dH(y, c) \le d - 1$}
        \BlackBox{$\GRSdecode$: a $\GRS$ decoder}
        \State{$\CC_{0} \gets \CC$}
        \For{$i = 1, \dots, d-1$}
            \State $f \gets \GRSdecode(\CC, y)$
            \label{line:rsdecode}
            \If {$f$ is not null}
                \State{\Return $\CC_{0}(f)$} 
                \label{line:ret}
            \Else 
                \State{puncture $\CC$ at coordinate $n - i + 1$}
                \label{line:punc}
                \State{$y \gets y[1..n - i]$} 
                \label{line:y}
            \EndIf
        \EndFor
    \end{algorithmic}
\end{algorithm}

\autoref{thm:grs-cover} below proves the correctness of \autoref{alg:GRS-cover}. 

\begin{theorem}
    \label{thm:grs-cover}
    Given an $[n, k, d]_{q}$ GRS code $\CC$ and $y \in \fq^{n}$, $\GRScover(\CC, y)$ returns a codeword $c \in \CC$ with $\dH(y, c) \le d - 1$. 
\end{theorem}

\begin{proof}
    Puncturing an $[n, k]$ GRS code at any coordinate gives an $[n - 1, k]$ GRS code. 
    Denote by $y^{(i)}$ and $\CC_{i}$ the 
    values of $y$ and $\CC$ on lines \ref{line:y} and \ref{line:punc}, respectively, at step $i$. 
    Plainly,
    for any $i$, 
    \begin{align}
        \dH(y^{(i+1)}, \CC_{i + 1}(f)) 
        &\le \dH(y^{(i)}, \CC_{i}(f)). 
        \label{eq:reduction}
    \end{align}
    If $\GRSdecode(\CC_{i}, y^{(i)})$ 
    is successful, then 
    by repeated application of \eqref{eq:reduction}, 
    \begin{align*}
        \dH(y^{(i)}, \CC_{i}(f)) \le \dH(y, \CC(f)) \le d - 1. 
    \end{align*}
    It remains to show that $\GRSdecode(\CC_{i}, y^{(i)})$ does indeed succeed for some $i < d$. 
    To see this, note that, 
    for each $i$, 
    \begin{align}
        d_{\min}(\CC_{i + 1}) = d_{\min}(\CC_{i}) - 1.
        \label{eq:dmin-reduction}
    \end{align}
    Hence, $\GRSdecode(\CC_{i}, y^{(i)})$ will succeed in at most $d - 1$ steps, since $d_{\min}(\CC_{i}) = 1$ when $i = d - 1$, at which point any $y^{(i)}$ becomes a valid codeword.
\end{proof}

We emphasize that \autoref{thm:grs-cover} does not claim that \autoref{alg:GRS-cover} returns the closest codeword to the input vector, but rather guarantees returning of a codeword within the covering radius of the code (there might be multiple codewords within the covering radius of the input vector). This distinction is essential: finding the closest codeword (i.e., solving the minimum-distance decoding problem) is NP-hard in general, whereas the proposed covering algorithm achieves its guarantee with polynomial-time complexity. A detailed analysis of the complexity of \autoref{alg:GRS-cover} will be provided later in this section. 

\subsection{Improvement Using List Decoding}

Consider the following modification of \autoref{alg:GRS-cover}, where $\GRSdecode$ is the Guruswami--Sudan list decoder (GS) \cite{Guruswami06} and returns the best candidate in the list (i.e., the codeword with the shortest Hamming distance from $y$). 
There are several advantages of doing this. 
First, it was shown by McEliece \cite{McEliece03, McEliece03-1} that GS almost always returns a list of size $1$, whereby implying that it is essentially a unique decoder. 
Next, the decoding radius of GS is given by: 
\begin{align}
    \label{eq:tauGS}
    \tau_{\GS} 
    &:= n - 1 - \lfloor \sqrt{(k - 1) n} \rfloor. 
\end{align}
This implies, in particular, that \autoref{alg:GRS-cover} will now return a $\GRS$ codeword $c$ within distance $\tau_{\GS}$ of $y$, i.e., with $\dH(y, c) \le \tau_{\GS}$, if such a codeword exists~\cite{McEliece03-1}. 
This potentially reduces the average number of punctures, thereby improving the average covering radius. 

\hide{
\begin{algorithm}[!htbp]
    \caption{$\GRScover(\CC, y)$} 
    \label{alg:GRS-cover-2}
    \begin{algorithmic}[1]
        \Require{$[n, k, d]_{q}$ $\GRS$ code $\CC$, arbitrary vector $y \in \fq^{n}$}
        \Ensure{A codeword $c \in \CC$ with $\dH(y, c) \le d - 1$}
        \BlackBox{$\GRSdecode$: a $\GRS$ decoder}
        \State{$\CC_{0} \gets \CC$}
        \For{$i = 1, \dots, d-1$}
            \State{{\tt min} $\gets d - 1$}
            \State{{\tt ret} $\gets$ {\tt NUL}}
            \State $L \gets \GRSdecode(\CC, y)$
            \label{line:rsdecode}
            \If {$L \neq \emptyset$}
                \For {$f \in L$}
                    \State{{\tt cur} $\gets C(f)$}
                    \State{{\tt temp} $\gets \dH(${\tt cur}$, y)$}
                    \If {{\tt temp} $<$ {\tt min}}
                        \State{{\tt ret} $\gets$ {\tt cur}}
                        \State{{\tt min} $\gets$ {\tt temp}}
                    \EndIf
                \EndFor
                \State{\Return ret} 
            \Else 
                \State{puncture $\CC$ at coordinate $n - i + 1$}
                \label{line:punc}
                \State{$y \gets y[1..n - i]$} 
                \label{line:y}
            \EndIf
        \EndFor
    \end{algorithmic}
\end{algorithm}
}

\subsection{Complexity Analysis}

Plainly, the worst-case running time of \autoref{alg:GRS-cover} is $(n - k) \cdot T(n)$, where $T(n)$ is the running time of $\GRSdecode$. 

\begin{itemize}
    \item For unique decoding, algorithms like Berlekamp--Welch (BW) are known with $T(n) = O(n^{3})$~\cite[Theorem~15.1.4]{ecc} as well as more efficient ones with $T(n) = O(n^{2})$ and $T(n) = O(n \polylog(n))$~\cite[\S~15.4]{ecc}. 
    \item For list decoding, there are known implementations of $\GS$ with $T(n) = O(s^{4} n^{2})$~\cite[\S~VII]{McEliece03}, \cite[\S~VII]{McEliece03}, \cite{Shokrollahi00}, \cite[\S~V]{Roth00},\hide{solving the interpolation problem 
    in GS takes $O(n^{3})$ time using the Feng--Tzeng algorithm \cite[\S~VIII]{McEliece03} or naive Gaussian elimination. This was improved by K\"{o}tter to $O(s^{4} n^{2})$ \cite[\S~VII]{McEliece03}, which is the most efficient known solution~
    \cite[\S~4]{McEliece03-1}. 
    
    The most efficient known solution to the factorization problem 
    of GS is due to Gao and Shokrollahi \cite{Shokrollahi00} with a time complexity of $O(\ell^{3} k^{2})$, 
    although the Roth--Ruckenstein algorithm \cite[\S~V]{Roth00} is quite competetive \cite[\S~4]{McEliece03-1} 
    (see also \cite[\S~IX]{McEliece03}). 
    Here $\ell$ is a design parameter, typically a small constant \cite{Roth00}, which is an upper bound on the size of the list of decoded codewords, which is bounded above by the degree of the interpolation polynomial in the second variable.
    ~\cite{Guruswami06}. 
    
    In general, $\ell \le \lfloor c / (k - 1) \rfloor
    = O(s \sqrt{n / k})$ (see also \cite[Theorem~4.8]{Guruswami06}), which gives $O(\ell^{3} k^{2}) = O(s^{3} n \sqrt{k n})$, yielding a $O(s^{4} n^{2})$ time complexity for \autoref{alg:GRS-cover}.}
    where $s > 0$ is an integer parameter known as the \emph{interpolation multiplicity}\hide{, $\GS(y, s)$ 
    returns all $f \in \CF(k, q)$ such that $\dH(y, \GRS(f)) \le n - \tau_{s}$, 
    where $\tau_{s} := \lfloor c / s \rfloor + 1$ and $c := \lfloor \sqrt{(k - 1) n s (s + 1)} \rfloor$}. 
\end{itemize}


\subsection{Average Number of Punctures}

In practical settings, the average-case running time of an algorithm is generally of more interest rather than the worst-case. Note that the average-case running time of \autoref{alg:GRS-cover} depends on the average number of punctures needed for $\GRSdecode$ to succeed. 
Based on the data in \autoref{tab:punctures}, we made the following conjecture in an earlier version of this work \cite{riasat2025covering}. 

\begin{conjecture}
    \label{conj:punctures}
    For an $[n, k, d]_{q}$ GRS code, the average number of punctures needed for \autoref{alg:GRS-cover} to succeed in returning a codeword within the covering radius is less than 
    \begin{itemize}
        \item $c_{1} (d - 1)$ for $\GRSdecode$ a unique decoder, and 
        \item $c_{2}$ for $\GRSdecode$ the $\GS$ list decoder, 
    \end{itemize}
    for some positive constants $c_{1} <1$ and $c_2$. 
\end{conjecture}
However, it turns out that the correct behavior of \autoref{alg:GRS-cover} is given by the following theorem. 

\begin{theorem}
    \label{thm:punctures-list}
    Suppose that $\GRSdecode$ in \autoref{alg:GRS-cover} is a list decoder that returns all codewords within its decoding radius. Then, averaged uniformly on all inputs $y \in \fqn$, the number $N_{{\rm punc}}$ of punctures needed for $\GRScover(\CC, y)$\hide{ takes $d - 2$ punctures on average to find a codeword within the covering radius for sufficiently large $k$.} to succeed for an $[n, k, d]_{q}$ GRS code $\CC$ satisfies the following inequalities:
    \begin{multline}
        \label{eq:punc-list}
        d - 1 - \frac{1}{q^{d - 1}} 
        \sum_{i = 0}^{d - 2} 
        q^{i} \Vol_{q}(\tau_{i}, n - i) 
        \le \overline{N}_{{\rm punc}} 
        \\ 
        \le d - 1 - \frac{1}{q^{d - 1}} 
        \sum_{i = 0}^{d - 2} 
        \frac{q^{i} \Vol_{q}(\tau_{i}, n - i)}{L_{i}},
    \end{multline} 
    where 
    $\tau_{i}$ and $L_{i}$ are respectively the decoding radius and the maximum list size of $\GRSdecode(\CC_{i}, \cdot)$. 
\end{theorem}

\begin{proof}
    Using the notation from the proof of \autoref{thm:grs-cover}, let 
    \begin{align}
        \label{eq:Sdef}
        S_{i} 
        &:= \{(y, c) \in \fqn \times \CC_{i}: \dH(y^{(i)}, c) \le \tau_{i}\}, \\ 
        \label{eq:Ldef}
        \ell_{i}(y) 
        &:= \#\{c \in \CC_{i}: \dH(y^{(i)}, c) \le \tau_{i}\}. 
    \end{align}
    In other words, $S_{i}$ is the set of all pairs $(y, c)$ such that the output list of $\GRSdecode(\CC_{i}, y^{(i)})$ contains a polynomial $f$ with $c = \CC_{i}(f)$\hide{, after puncturing the last $i$ coordinates of $y$ lands in $B(c, \tau_{i})$}, and $\ell_{i}(y)$ is the  \hide{number of codewords in $\CC_{i}$ that are within distance $\tau_{i}$ of this punctured vector (i.e., the }size of the list\hide{ of $\GRSdecode(\CC_{i}, y^{(i)})$}. 
    Then, the number of vectors $y \in \fqn$ for which \autoref{alg:GRS-cover} succeeds within $i$ punctures is given by: 
    \begin{align}
        \label{eq:VL}
        V_{i} 
        &:= \#\{y \in \fqn: \ell_{i}(y) \ge 1\}. 
    \end{align}
    In particular, the average number of punctures $\overline{N}_{{\rm punc}}$ needed for \autoref{alg:GRS-cover} to succeed is given by: 
    \begin{align}
        \label{eq:avg-punc}
        \frac{1}{q^{n}} 
        \sum_{i = 1}^{d - 1} 
        i (V_{i} - V_{i - 1}) 
        &= \frac{1}{q^{n}} \left((d - 1) V_{d - 1} - \sum_{i = 0}^{d - 2} V_{i}\right) \notag \\ 
        &= d - 1 - \frac{1}{q^{n}} \sum_{i = 0}^{d - 2} V_{i},
    \end{align}
    since $V_{d - 1} = q^{n}$. 
    Now, by \eqref{eq:Sdef} and \eqref{eq:Ldef}, 
    \begin{align}
        |S_{i}| 
        &= \sum_{y \in \fqn} 
        \ell_{i}(y) 
        = \sum_{c \in \CC_{i}} 
        q^{i} |B(c, \tau_{i})| \notag \\ 
        &= q^{k + i} \Vol_{q}(\tau_{i}, n - i). 
        \label{eq:Si1}
    \end{align} 
    Since $L_{i}$ is the maximum list size of $\GRSdecode(\CC_{i}, \cdot)$, 
    \begin{align}
        \label{eq:Lchain}
        1_{\{\ell_{i}(y) \ge 1\}} 
        &\le \ell_{i}(y) 
        \le 1_{\{\ell_{i}(y) \ge 1\}} L_{i}. 
    \end{align}
    Summing \eqref{eq:Lchain} over $y \in \fqn$ using \eqref{eq:VL} and \eqref{eq:Si1} 
    gives
    \begin{align}
        \label{eq:Schain}
        V_{i} 
        &\le |S_{i}| 
        \le V_{i} 
        L_{i}, 
    \end{align}
    i.e., 
    \begin{align}
        \label{eq:Vchain}
        \frac{|S_{i}|}{L_{i}} 
        &\le V_{i} 
        \le |S_{i}|. 
    \end{align}
    Finally, \eqref{eq:punc-list} follows from \eqref{eq:avg-punc}, \eqref{eq:Si1} and \eqref{eq:Vchain}. 
\end{proof}


If $\GRSdecode$ is a unique decoder, we can substitute $L_{i} = 1$ for all $i$ into \autoref{thm:punctures-list} and obtain the following corollary. 

\begin{corollary}
    \label{cor:punctures}
    Suppose that $\GRSdecode$ in \autoref{alg:GRS-cover} is a unique decoder with maximum decoding radius\hide{ that succeeds if and only if the input is within its decoding radius}. Then, averaged uniformly over all inputs $y \in \fqn$, the number $N_{{\rm punc}}$ of punctures needed for $\GRScover(\CC, y)$\hide{ takes $d - 2$ punctures on average to find a codeword within the covering radius for sufficiently large $k$.} to succeed for an $[n, k, d]_{q}$ GRS code $\CC$ is given by: 
    \begin{align}
        \label{eq:punc-unique}
        \overline{N}_{{\rm punc}} 
        &= d - 1 - \frac{1}{q^{d - 1}} 
        \sum_{i = 0}^{d - 2} 
        q^{i} \Vol_{q}(\tfrac{n - k - i}{2}, n - i). 
    \end{align}
    In particular,  
    \begin{align}
        \label{eq:punc}
        (d - 1) \left(1 - \frac{1}{q}\right) 
        &\le \overline{N}_{{\rm punc}} \le 
        d - 1
        ,
    \end{align} 
    i.e., $\overline{N}_{{\rm punc}} / (d - 1) \to 1$ as $q \to \infty$.
    \hide{
    Furthermore, 
    \begin{align}
        \label{eq:punc2}
        d - 1 - \frac{2^{k} \varvarepsilon (1 - \varvarepsilon^{n - k})}{1 - \varvarepsilon} 
        &\le \overline{N}_{{\rm punc}} \le d - 1,
    \end{align}
    where 
    \begin{align}
        \label{eq:eps}
        \varvarepsilon 
        &= \frac{2 \sqrt{q - 1}}{q}. 
    \end{align}
    }
    \hide{Furthermore, 
    \begin{align*}
        N \ge 
        d - 1 - 
        \frac{q^{k H_{q}(\frac{1 - R}{2})} - q^{k - n + n H_{q}(\frac{1 - R}{2})}}{q^{1 - H_{q}(\frac{1 - R}{2})} - 1},
    \end{align*}
    where $R = k / n$. 
    }
\end{corollary}

\begin{proof}
    \hide{Let $y$ denote the input vector. Suppose that 
    $$y \in \bigcup_{c \in \CC} B(c, \tau),$$ 
    where $\tau$ is the decoding radius of $\GRSdecode$. Then $\GRScover$ succeeds before any puncturing occurs, and the number of such $y$'s is at least 
    $$\resizebox{\hsize}{!}{$\displaystyle q^{k} \Vol_{q}(\tau, n) - \frac{q^{n + k - d}}{2}  \sum_{w = d}^{2 \tau} \binom{n}{w} [\Vol_{q}(\tau, w) - \Vol_{q}(w - \tau - 1, w)]$}$$
    by \autoref{cor:union-lower-bound}. Using the notation from the proof of \autoref{thm:punctures-list}, }
    Taking $\tau_{i} = (n - k - i) / 2$ and $L_{i} = 1$ for all $i$ in \eqref{eq:punc-list} gives \eqref{eq:punc-unique}, and in \eqref{eq:Schain} gives,  
    by \eqref{eq:Si1}, 
    \begin{align}
        \label{eq:Vi}
        V_{i} 
        &= |S_{i}| 
        = q^{k + i} 
        \Vol_{q}(\tau_{i}, n - i). 
    \end{align}
    Thus, 
    \hide{where 
    \begin{align*}
        \tau_{i} 
        &= \left\lfloor\frac{d - 1 - i}{2}\right\rfloor 
        = \left\lfloor\frac{n - k - i}{2}\right\rfloor. 
    \end{align*}
    Since $\tau_{i} / (n - i) \le 1 - 1 / q$, 
    \begin{align*}
        \Vol_{q}(\tau_{i}, n - i) 
        &\le q^{(n - i) H_{q}(\frac{\tau_{i}}{n - i})} 
    \end{align*}
    and 
    \begin{align*}
        \Vol_{q}(\tau_{i}, n - i) 
        &\ge q^{(n - i) H_{q}(\frac{\tau_{i}}{n - i}) - C \log(n - i)}
    \end{align*}
    for some constant $C > 0$ and sufficiently large 
    $n - i \in [k + 1, n]$. Thus, 
    \begin{align*}
        \sum_{i = 0}^{d - 2} 
        \frac{V_{i}}{q^{n}} 
        &\ge 
        \frac{V_{d - 2}}{q^{n}} \\ 
        &= 
        \frac{q^{k + d - 2} \Vol_{q}(\tau_{d - 2}, n - d + 2)}{q^{n}} \\ 
        &= \frac{\Vol_{q}(\tfrac{1}{2}, k + 1)}{q} \\ 
        &= \frac{1}{q}. 
    \end{align*}
    i.e. 
    \begin{align*}
        \frac{V_{i}}{q^{n}} 
        &\sim q^{k - n + n H_{q}(\frac{\tau_{i}}{n - i})) + i (1 - H_{q}(\frac{\tau_{i}}{n - i}))}.
    \end{align*} 
    On the other hand, }
    \begin{align}
        \notag 
        \frac{1}{q^{n}} 
        \sum_{i = 0}^{d - 2} V_{i}
        &\le \frac{(d - 1) V_{d - 2}}{q^{n}} \\ 
        \notag 
        &\le \frac{(d - 1) q^{k + d - 2} \Vol_{q}(\tau_{d - 2}, n - d + 2)}{q^{n}} \\ 
        \notag 
        &= \frac{(d - 1) \Vol_{q}(\tfrac{1}{2}, k + 1)
        }{q} \\ 
        &= \frac{d - 1}{q}
        \label{eq:sum-Vi}
        , 
    \end{align}
    \hide{i.e., taking 
    \begin{align*}
        c_{1} 
        &= \max\begin{cases}
            1 - q^{n (1 - R - H_{q}((1 - R) / 2)) - \varvarepsilon} & n \le N_{0}(\varvarepsilon) \\ 
            1 - q^{-N_{0}(\varvarepsilon) (R - 1 + H_{q}((1 - R) / 2)) - \varvarepsilon} 
        \end{cases}
    \end{align*}
    which follows from noting that 
    \begin{align*}
        &\frac{q^{- (n - i) (1 - H_{q}(\tfrac{\tau_{i}}{n - i}))}}{q^{- (n - i + 1) (1 - H_{q}(\tfrac{\tau_{i - 1}}{n - i + 1}))}} \\ 
        &= q^{1 + (n - i) H_{q}(\frac{\tau_{i}}{n - i}) - (n - i + 1) H_{q}(\frac{\tau_{i - 1}}{n - i + 1})} \\ 
        &= q^{1 + (n - i) H_{q}(\frac{1}{2} - \frac{k}{2 (n - i)}) - (n - i + 1) H_{q}(\frac{1}{2} - \frac{k}{2 (n - i + 1)})} \\ 
        &\approx q. 
    \end{align*}
    which is guaranteed whenever 
    \begin{align*}
        \frac{2 \sqrt{e (q - 1)}}{q} 
        < \frac{\sqrt{(n - i - 1)^{2} - k^{2}}}{n - i}. 
    \end{align*}
    For small $i$, this is guaranteed for 
    \begin{align*}
        \frac{2 \sqrt{q - 1}}{q} 
        < \frac{\sqrt{(n - i - 1)^{2} - k^{2}}}{n - i} \\ 
        \impliedby \frac{2 \sqrt{q - 1}}{q} 
        < \sqrt{1 - R^{2}}, 
    \end{align*}
    which holds for sufficiently large $q$. 
    For $i$ close to $n - k - 1 = d - 2$, this is guaranteed for 
    \begin{align*}
        \frac{2 \sqrt{e (q - 1)}}{q} 
        < \frac{\sqrt{(n - i - 1)^{2} - k^{2}}}{n - i} \\ 
        \iff \frac{2 \sqrt{e (q - 1)}}{q} 
        < \sqrt{1 - R^{2}}, 
    \end{align*}
    which holds for sufficiently large $q$. 
    using the dominant term in the sum, as desired. 
    }since the $V_{i}$ are non-decreasing\hide{ by \eqref{eq:VL}}. Substituting \eqref{eq:sum-Vi} into \eqref{eq:avg-punc} gives \eqref{eq:punc}. 
\end{proof}

\autoref{cor:punctures} clearly shows that the first part of \autoref{conj:punctures} does not hold in general. 
We believe that a similar analysis that carefully takes into account intersections of Hamming spheres (e.g., see \autoref{thm:union-lower-bound} below) can also help resolve the second part of the conjecture. 
\hide{A similar analysis gives the following result for the list decoder case.
Based on \autoref{fig:covering-list}, we propose the following modification to the second part. 

\begin{conjecture}
    For an $[n, k, d]_{q}$ GRS code $\CC$ of fixed rate $R$, the average number of punctures needed for \autoref{alg:GRS-cover} to succeed with $\GRSdecode = \GS$ is $\Theta(\rhoh(\CC))$. 
\end{conjecture}
\begin{corollary}
    \label{cor:punctures-list}
    Suppose that $\GRSdecode$ in \autoref{alg:GRS-cover} is the $\GS$ list decoder with maximum decoding radius. Then, averaged over all inputs $y \in \fqn$, the number $N_{{\rm punc}}$ of punctures needed for $\GRScover(\CC, y)$\hide{ takes $d - 2$ punctures on average to find a codeword within the covering radius for sufficiently large $k$.} to succeed for an $[n, k, d]_{q}$ GRS code $\CC$ satisfies: 
    \begin{align}
        \label{eq:punc-gs}
        (d - 1) \left(1 - \frac{1}{q}\right) 
        &\le \overline{N}_{{\rm punc}}^{({\rm list})} \le 
        d - 1
        ,
    \end{align} 
    i.e., $\overline{N}_{{\rm punc}} / (d - 1) \to 1$ as $q \to \infty$.
\end{corollary}

\begin{proof}
    Substituting $L_{i} \le (n - i) / (k - 1) = \Theta(1 / R)$ and the trivial bound $\Vol_{q}(\tau_{i}, n - i) \ge 1$ into \eqref{eq:punc-list} gives 
    \begin{align}
        N_{{\rm punc}}^{({\rm list})} 
        &\le d - 1 - \frac{\Theta(R)}{q^{d - 1}} \sum_{i = 0}^{d - 2} q^{i} 
        = d - 1 - \Theta\left(\frac{R}{q}\right). 
        \label{eq:punc-gs-upper}
    \end{align}
    In the other direction, we use the classical bound 
    \begin{align}
        \label{eq:vol-upper}
        \Vol_{q}(\tau_{i}, n - i) 
        &\le q^{(n - i) H_{q}(\frac{\tau_{i}}{n - i})} 
        \le q^{(n - i) H_{q}(1 - \sqrt{R})}, 
    \end{align}
    since $\tau_{i} = n - i - 1 - \lfloor \sqrt{(k - 1) (n - i)} \rfloor$ and $H_{q}$ is increasing in $[0, 1 - 1 / q]$. 
    Substituting \eqref{eq:vol-upper} into \eqref{eq:punc-list} gives 
    \begin{align}
        N_{{\rm punc}}^{({\rm list})} 
        &\ge d - 1 - \frac{1}{q^{d - 1}} \sum_{i = 0}^{d - 2} q^{i + (n - i) H_{q}(1 - \sqrt{R})} 
        \notag \\ 
        &\ge (d - 1) (1 - q^{n H_{q}(1 - \sqrt{R}) + (d - 2) (1 - H_{q}(1 - \sqrt{R}))}) 
        \notag 
    \end{align}
\end{proof}}

\subsection{Coverage Fraction} 

Next, we provide a lower bound on the fraction of space covered by Hamming spheres of a given radius centered at the codewords of an MDS code in terms of its weight and intersection distribution. 
The key observation is that one can explicitly compute the sizes of the Hamming spheres and their pairwise intersections, as \autoref{thm:union-lower-bound} below demonstrates. 


\begin{theorem}
    \label{thm:union-lower-bound}
    Let $\CC$ be an $[n, k, d]_{q}$ MDS code. Then, the fraction of the space $\fqn$ covered by Hamming spheres of radius $\tau$ associated with $\CC$ is lower bounded as follows: 
    \begin{multline}
        \label{eq:coverage-fraction}
        \frac{1}{q^{n}} 
        \bigabs{\bigcup_{c \in \CC} B(c, \tau)} \\ 
        \ge \frac{1}{q^{d - 1}} \left(\sum_{i = 0}^{\tau} \binom{n}{i} (q - 1)^{i} 
        - \frac{1}{2} \sum_{w = d}^{2 \tau} A_{w} I(w, \tau)\right), 
    \end{multline}
    where the \emph{weight distribution} $A_{w}$ is given by:
    \begin{align}
        &A_{w} 
        := \binom{n}{w} 
        \sum_{j = 0}^{w - d} (-1)^{j} 
        \binom{w}{j} (q^{w - d + 1 - j} - 1),  \label{eq:weight-distribution} 
    \end{align} 
    and the \emph{intersection distribution} $I(w, \tau)$ is given by: 
    \begin{align}
        & I(w, \tau) := 
        \label{eq:intersection} \\ 
        &\hspace{-1mm} 
        \resizebox{\hsize}{!}{
        $\displaystyle \sum_{z = 0}^{n - w} 
        \binom{n - w}{z} 
        (q - 1)^{n - w - z}
        \hspace{-3mm}
        \sum_{
            \substack{
                n - \tau - z \le u, v \le \tau  
                \\ 
                u + v \le w 
            }
        } \binom{w}{u} 
        \binom{w - u}{v}
        (q - 2)^{w - u - v}$.} 
        \notag 
    \end{align} 
\end{theorem}

\begin{proof}
    By the inclusion-exclusion principle, 
    \begin{multline}
        \bigabs{\bigcup_{c \in \CC} B(c, \tau)} 
        \ge \sum_{c \in \CC} |B(c, \tau)| - \frac{1}{2} \sum_{c_{1} \neq c_{2}} |B(c_{1}, \tau) \cap B(c_{2}, \tau)| \\ 
        \label{eq:union-lb-proof}
        = |\CC| \left(\sum_{i = 0}^{\tau} \binom{n}{i} (q - 1)^{i} 
        - \frac{1}{2} \sum_{w = d}^{2 \tau} A_{w} I(w, \tau)\right), 
    \end{multline}
    where $A_{w}$ is the number of codewords $c \in \CC$ of weight $w$, and $I(w, \tau) := |B(c_{1}, \tau) \cap B(c_{2}, \tau)|$ for $c_{1}, c_{2} \in \CC$ with $\dH(c_{1}, c_{2}) = w$. 
    Now, \eqref{eq:weight-distribution} is known to hold for $d \le w \le n$ (e.g., see~\cite{Macwilliams77}), and \eqref{eq:intersection} follows by a counting argument, where $u$ (resp. $v$) represents the number of indices in $\{1, \dots, n\}$ where $y$ and $c_{1}$ (resp. $c_{2}$) agree, and $z$ represents the number of indices where $c_{1}$, $c_{2}$ and $y$ agree. 
    Substituting $|\CC| = q^{k}= q^{n - d + 1}$ into \eqref{eq:union-lb-proof} gives \eqref{eq:coverage-fraction}. 
\end{proof}



\hide{\autoref{thm:union-lower-bound} gives the following bounds on the fraction of space covered by the union of the Hamming spheres. 


\begin{corollary}
    \label{cor:union-lower-bound} 
    Let $\CC$ be an $[n, k, d]_{q}$ GRS code. Then
    the fraction of the ambient space covered by the union of Hamming spheres of radius $\tau$ centered at the codewords satisfies 
    \begin{align*}
        & \resizebox{\hsize}{!}{$\displaystyle q^{k - n} 
        \Vol_{q}(\tau, n) 
        - \frac{q^{k - d}}{2}  
        \sum_{w = d}^{2 \tau} 
        \binom{n}{w} 
        [\Vol_{q}(\tau, w) - \Vol_{q}(w - \tau - 1, w)]$} 
        \\ 
        &\le 
        q^{-n} \bigabs{\bigcup_{c \in \CC} B(c, \tau)} 
        \le q^{k - n} 
        \Vol_{q}(\tau, n),  
    \end{align*}
    where (see, e.g., \cite{ecc})
    \begin{align*}
        \Vol_{q}(\tau, n) 
        &:= |B(0, \tau)| 
        = \sum_{j = 0}^{\tau} 
        \binom{n}{r} (q - 1)^{j} 
        \approx q^{n H_{q}(\tau / n)}
    \end{align*}
    for $H_{q}(x) := x \log_{q}(q - 1) - x \log_{q}(x) - (1 - x) \log_{q}(1 - x)$. 
\end{corollary}

\begin{proof}
    Observe that 
    \begin{align*}
        A_{w} 
        &\le \binom{n}{w} q^{w - d} 
    \end{align*}
    by \eqref{eq:weight-distribution}, 
    and 
    \begin{align*}
        & I(w, \tau) 
        \ \\ 
        &\resizebox{\hsize}{!}{$\displaystyle = \sum_{z = 0}^{n - w} 
        \binom{n - w}{z} 
        (q - 1)^{n - w - z}
        \sum_{
                n - \tau - z \le u \le \tau 
        } \binom{w}{u} 
        \sum_{
                n - \tau - z \le v \le \tau 
        } 
        \binom{w - u}{v}
        (q - 2)^{w - u - v}$} \\ 
        & 
        \le \sum_{z = 0}^{n - w} 
        \binom{n - w}{z} 
        (q - 1)^{n - w - z}
        \sum_{
                n - \tau - z \le u \le \tau 
        } \binom{w}{u} 
        (q - 1)^{w - u}
        \\ 
        &\resizebox{\hsize}{!}{$\displaystyle = \sum_{z = 0}^{n - w} 
        \binom{n - w}{z} 
        (q - 1)^{n - w - z} 
        [\Vol_{q}(\tau, w) - \Vol_{q}(n - \tau - z - 1, w)]$} \\ 
        &\le 
        q^{n - w} 
        [\Vol_{q}(\tau, w) - \Vol_{q}(w - \tau - 1, w)] 
    \end{align*}
    by \eqref{eq:intersection}. The conclusion now follows from \autoref{thm:union-lower-bound}. 
\end{proof}}

Note that \autoref{thm:union-lower-bound} is sharp when overlaps of three or more Hamming spheres are negligible, which is the case when $\tau$ is sufficiently small, and, in particular, expected to be the case for the smallest $\tau$ such that the union of Hamming spheres of radius $\tau$ covers most of the space. 
We \hide{estimated this bound for $\tau \in (d / 2, d)$ for several $[n, k, d]_{q}$ GRS codes and }observe numerically in \autoref{sec:taumax} that the radius $\tau_{\max}$ yielding the best lower bound is generally close 
to $\tau_{\GS}$ defined in \eqref{eq:tauGS}. 
\hide{In the next section, we will show numerically that 
\begin{itemize}
    \item $\tau_{\GS} \approx \tau_{\max}$ for all $1 \le k < n$ with $(q, n) = (47, 46)$, 
    implying that Hamming spheres of radius $\tau_{\GS}$ cover a significant fraction of the ambient space,
    \item \autoref{conj:punctures} is satisfied for all $1 \le k < n$ with $(q, n) = (7, 6)$, $(11, 10)$, 
    and
    \item the average covering radius obtained via $\GRScover$ using the $\GS$ list decoder is very close to that obtained using a \emph{maximum-a-posteriori} (MAP) decoder for $(q, n) = (7, 6)$ and $1 \le k < n$.
\end{itemize}}

\subsection{Simulation Results and Comparison with Random Coding Bound}
\label{sec:simulation-grs}

\subsubsection{Fixed Length}


\autoref{tab:punctures} lists the empirical average number of punctures for $\GRScover$ to succeed on an $[n, k]_{q}$ GRS code $\CC$ with $\GRSdecode$ a unique decoder (BW) and a list decoder (GS), 
for $(q, n) \in \{(7, 6), (11, 10), (17, 16)\}$, and $1 \le k < n$. For each $k$, the average was computed over $500$ Monte Carlo simulations. 
It shows that the number of punctures required for \autoref{alg:GRS-cover} to succeed using BW is close to the worst-case covering radius of the code, but is much smaller when using the GS list decoder. 
In other words, while Hamming spheres of radius $< d / 2$ do not cover the space well, those of radius $\tau_{\GS}$ cover a large fraction of the space. 

\begin{table}[!htbp]
    \begin{minipage}[t]{0.32\linewidth}
        \centering
        \begin{tabular}[t]{ccc}
            \toprule 
            $k$ & BW & GS \\
            \midrule 
            1 & 3.836 & 0 \\
            2 & 2.44 & 0.006 \\
            3 & 1.654 & 0.08 \\
            4 & 0.532 & 0.264 \\
            5 & 0.868 & 0 \\
            
            \bottomrule
    \end{tabular}
    \medskip
    
    $q = 7$, $n = 6$
    \end{minipage}%
    \begin{minipage}[t]{0.33\linewidth}
        \centering
        \begin{tabular}[t]{ccc}
            \toprule 
            $k$ & BW & GS \\
            \midrule 
            1 & 8.224 & 0 \\
            2 & 6.872 & 0.452 \\
            3 & 5.654 & 0.366 \\
            4 & 4.34 & 0.378 \\
            5 & 3.022 & 0.7 \\
            6 & 1.864 & 0.744 \\ 
            7 & 1.428 & 0.036 \\
            8 & 0.376 & 0.16 \\
            9 & 0.922 & 0 \\
            \bottomrule
        \end{tabular}
        \medskip
        
        $q = 11$, $n = 10$
    \end{minipage}
    \begin{minipage}[t]{0.34\linewidth}
        \centering
        \begin{tabular}[t]{ccc}
            \toprule 
            $k$ & BW & GS \\
            \midrule 
            1 & 14.573 & 0 \\
            2 & 13.293 & 0.69 \\
            3 & 12.253 & 1.65 \\
            4 & 10.927 & 0.95 \\
            5 & 9.847 & 2.097 \\
            6 & 8.493 & 0.837 \\ 
            7 & 7.22 & 0.61 \\
            8 & 6.18 & 1.027 \\
            9 & 4.843 & 0.76 \\
            10 & 3.66 & 1.173 \\ 
            11 & 2.717 & 0.2 \\ 
            12 & 1.587 & 0.627 \\ 
            13 & 1.187 & 0.007 \\ 
            14 & 0.18 & 0.073 \\ 
            15 & 0.947 & 0 \\ 
            \bottomrule
        \end{tabular}
        \medskip
        
        $q = 17$, $n = 16$
    \end{minipage}
    \caption{Average number of punctures before $\GRScover$ succeeds using unique (BW) and list (GS) decoders} 
    \label{tab:punctures}
\end{table}


We also estimated the average covering radius $\rhohbar(\CC)$ of $\CC$ by computing the average value of $\hide{\min_{c \in \CC} }\dH(y, c)$\hide{covering radius $\rhohbar(\CC)$} over $500$ simulations for $(q, n) = (7, 6)$ and $1 \le k < n$ with $\GRSdecode$ a unique decoder (BW), a list decoder (GS), and a MAP decoder, where $c$ is the output codeword of $\GRScover(\CC, y)$. The results are compared in \autoref{fig:covering-radius} with the worst-case covering radius $\rhoh(\CC) = n - k$ and the random coding bound \eqref{eq:random-hamming} given by \autoref{thm:random-hamming}. 
It can be seen that \autoref{alg:GRS-cover} using GS list decoder improves over random codes for all $k$ (despite the fact that GRS codes are not ideal for covering) and the resulting average covering radius is close to that obtained using the MAP decoder. 
In fact, \autoref{alg:GRS-cover} using BW also improves over random codes for sufficiently large $k$. 
The decoding radius parameter $\tau$ of GS in each case was taken to be its maximum decoding radius $\tau_{\GS}$ defined in \eqref{eq:tauGS}. 
Note that the MAP decoder provides the optimal covering solution with an exponential complexity.

\begin{figure}[!htbp]
    \centering
    \begin{tikzpicture}
    
    \definecolor{black}{rgb}{0, 0.0, 0}%
    \definecolor{blue}{rgb}{0, 0, 1}%
    \definecolor{red}{rgb}{1, 0, 0}
    \definecolor{green}{rgb}{0, 0.5, 0}
    
    \begin{axis}[
        font=\small,
        width=0.8\linewidth,
        height=0.6\linewidth,
        scale only axis,
        xmin=0.4,
        xmax=5.7,
        xtick={0,1,2,3,...,5,6},
        xlabel={Dimension ($k$)},  
        ymin=0.3,
        ymax=5.6,
        ytick={0,1,2,3,...,5,6},
        ylabel={Average Covering Radius},
        ylabel near ticks,
        minor tick num=4,
        grid=both,
        grid style={
            line width=.1pt, 
            draw=gray!7
        },
        major grid style={
            line width=.2pt,
            draw=gray!50,
        },
        legend style={
            font=\scriptsize, 
            at={(1,1)},
            anchor=north east, 
            draw=black,
            fill=white,
            legend cell align=left, 
            cells={align=left},
            row sep=0.1pt
        }
    ]
    
    \addplot[
        color=black, 
        solid, 
        solid
    ] table[
        y={create col/linear regression={y=Y}}
    ] {
        X Y
        0 6
        6 0
    };
    \addlegendentry{Upper Bound \\ ($n - k$)
    };
    
    \addplot[
        color=black, 
        solid, 
        solid, 
        mark=*, 
        mark size=1pt,
        mark options={solid}
    ] table[row sep=crcr]{%
        1 3.90549325\\ 
        2 2.89474223\\ 
        3 2.08001020\\ 
        4 1.44971291\\ 
        5 0.871936580\\
    };
    \addlegendentry{Random Codes \\ 
    (\autoref{thm:random-hamming})};
    
    \addplot[
        color=black, 
        solid, 
        dashed, 
        line width=0.7pt,
        mark=+, 
        mark size=2.3pt,
        mark options={solid}
    ] table[row sep=crcr]{%
        1 4.064\\ 
        2 \fpeval{1533/500}\\ 
        3 \fpeval{557/250}\\ 
        4 \fpeval{309/250}\\ 
        5 \fpeval{421/500}\\
    };
    \addlegendentry{\autoref{alg:GRS-cover} (BW)};
    
    \addplot[
        color=black, 
        solid, 
        dash dot, 
        line width=0.7pt,
        mark=o, 
        mark options={solid}
    ] table[row sep=crcr]{%
        1 3.714\\ 
        2 \fpeval{1361/500}\\ 
        3 \fpeval{963/500}\\ 
        4 \fpeval{307/250}\\ 
        5 \fpeval{421/500}\\
    };
    \addlegendentry{\autoref{alg:GRS-cover} (GS)};
    
    \addplot[
        color=black, 
        solid, 
        solid, 
        mark=x, 
        mark size=2.3pt,
        mark options={solid, line width=0.7pt}
    ] table[row sep=crcr]{%
        1 3.718\\ 
        2 \fpeval{273/100}\\ 
        3 \fpeval{242/125}\\ 
        4 \fpeval{123/100}\\ 
        5 \fpeval{429/500}\\
    };
    \addlegendentry{\autoref{alg:GRS-cover} \\ (MAP)};
    
    \end{axis}
    
    \end{tikzpicture}
    \caption{Comparison of simulated average covering radii of a $[6, k]_{7}$ GRS code and random codes for $1 \le k < 6$}
    \label{fig:covering-radius}
\end{figure}
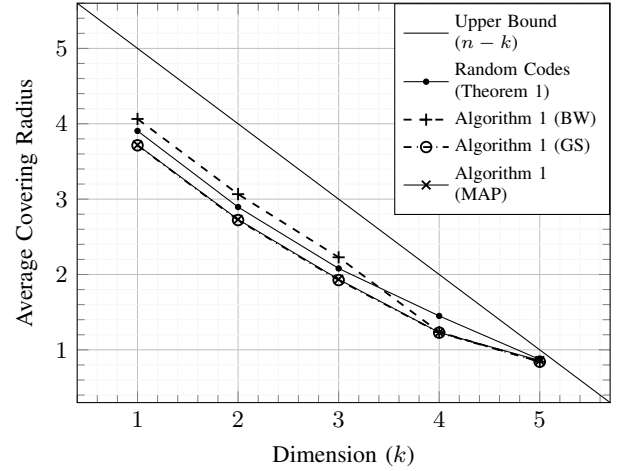

\subsubsection{Fixed Rate}

\autoref{fig:covering-list} shows the average covering radius of a $[n, k]_{q}$ GRS code $\CC$, where $n = q - 1$ and $k = \lfloor n R \rfloor$, for $5 \le q < 32$ given by \autoref{alg:GRS-cover} with $\GRSdecode = \GS$ \hide{as $\GRSdecode$ the Guruswami--Sudan list decoder }as well as the average number of punctures required by the algorithm to succeed, both computed over $300$ Monte Carlo simulations for $R \in \{1 / 3, 1 / 2, 2 / 3\}$. 
The list size parameter of $\GS$ in each case was taken to be the length of the code. 
This suggests that the average number of punctures and the average covering radius both increase with the covering radius $\rhoh(\CC) = n - k$ when the rate $R$ of the code is fixed and the decoding radius of $\GS$ is sub-optimal. 

\begin{figure*}[!htbp]
    \centering
    \subfigure[$R = 1 / 3$]{
        \begin{tikzpicture}
        
        \definecolor{black}{rgb}{0, 0.0, 0}%
        \definecolor{blue}{rgb}{0, 0, 1}%
        \definecolor{red}{rgb}{1, 0, 0}
        \definecolor{green}{rgb}{0, 0.5, 0}
        
        \begin{axis}[
            font=\footnotesize,
            width=0.4\linewidth,
            height=0.3\linewidth,
            scale only axis,
            xmin=0,
            xmax=32,
            xtick={0,5,10,15,...,35},
            xlabel={Length},  
            ymin=-2,
            ymax=23,
            ytick={-5,0,5,10,15,20,25},
            minor tick num=4,
            ylabel near ticks,
            grid=both,
            grid style={
                line width=.1pt, 
                draw=gray!7
            },
            major grid style={
                line width=.2pt,
                draw=gray!50,
            },
            legend style={
                font=\scriptsize, 
                at={(0,1.0)},
                anchor=north west, 
                fill=white,
                legend cell align=left
            },
            legend columns = 1
        ]
        
        \addplot [
            color=black, 
            solid, 
            mark=triangle, 
            mark size=3pt,
            mark options={solid}
        ]
            table[row sep=crcr]{%
            4 3\\
            6 4\\
            7 5\\
            8 6\\
            10 7\\
            12 8\\
            15 10\\
            16 11\\
            18 12\\
            22 15\\
            26 18\\
            28 19\\
            30 20\\
        };
        \addlegendentry{Worst-Case Covering Radius};
        
        \addplot [
            color=black, 
            solid, 
            mark=o, 
            mark options={solid}
        ]
            table[row sep=crcr]{%
            4 3\\
            6 \fpeval{277/100}\\
            7 \fpeval{217/60}\\
            8 \fpeval{22/5}\\
            10 \fpeval{751/150}\\
            12 \fpeval{579/100}\\
            15 \fpeval{2191/300}\\
            16 \fpeval{1253/150}\\
            18 \fpeval{688/75}\\
            22 \fpeval{1209/100}\\
            26 \fpeval{1138/75}\\
            28 \fpeval{2479/150}\\
            30 \fpeval{1733/100}\\
        };
        \addlegendentry{Average Covering Radius};
        
        \addplot [
            color=black, 
            solid, 
            mark=square, 
            mark options={solid}
        ]
            table[row sep=crcr]{%
            4 3\\
            6 0\\
            7 \fpeval{1/75}\\
            8 \fpeval{137/300}\\
            10 \fpeval{37/100}\\
            12 \fpeval{283/300}\\
            15 \fpeval{83/50}\\
            16 \fpeval{29/10}\\
            18 \fpeval{296/75}\\
            22 \fpeval{351/50}\\
            26 \fpeval{3269/300}\\
            28 \fpeval{4217/300}\\
            30 \fpeval{283/20}\\
        };
        \addlegendentry{Average Number of Punctures};
        
        \end{axis}
        
        \end{tikzpicture}
    } \quad 
    \subfigure[$R = 2 / 3$]{
        \begin{tikzpicture}
        
        \definecolor{black}{rgb}{0, 0.0, 0}%
        \definecolor{blue}{rgb}{0, 0, 1}%
        \definecolor{red}{rgb}{1, 0, 0}
        \definecolor{green}{rgb}{0, 0.5, 0}
        
        \begin{axis}[
            font=\footnotesize,
            width=0.4\linewidth,
            height=0.3\linewidth,
            scale only axis,
            xmin=0,
            xmax=32,
            xtick={0,5,10,15,...,35},
            xlabel={Length},  
            ymin=-2,
            ymax=23,
            ytick={-5,0,5,10,15,20,25},
            minor tick num=4,
            ylabel near ticks,
            grid=both,
            grid style={
                line width=.1pt, 
                draw=gray!7
            },
            major grid style={
                line width=.2pt,
                draw=gray!50,
            },
            legend style={
                font=\scriptsize, 
                at={(0,1.0)},
                anchor=north west, 
                fill=white,
                legend cell align=left
            },
            legend columns = 1
        ]
        
        \addplot [
            color=black, 
            solid, 
            mark=triangle, 
            mark size=3pt,
            mark options={solid}
        ]
          table[row sep=crcr]{%
          4 2\\
          6 2\\
          7 3\\
          8 3\\
          10 4\\
          12 4\\
          15 5\\
          16 6\\
          18 6\\
          22 8\\
          24 8\\
          26 9\\
          28 10\\
          30 10\\
        };
        \addlegendentry{Worst-Case Covering Radius};
        
        \addplot [
            color=black, 
            solid, 
            mark=o, 
            mark options={solid}
        ]
          table[row sep=crcr]{%
          4 1.25667\\
          6 1.20667\\
          7 1.96333\\
          8 \fpeval{293/150}\\
          10 \fpeval{269/100}\\
          12 \fpeval{199/75}\\
          15 \fpeval{463/150}\\
          16 \fpeval{599/150}\\
          18 \fpeval{1187/300}\\
          22 \fpeval{167/30}\\
          24 \fpeval{412/75}\\
          26 \fpeval{639/100}\\
          28 \fpeval{419/60}\\
          30 \fpeval{413/60}\\
        };
        \addlegendentry{Average Covering Radius};
        
        \addplot [
            color=black, 
            solid, 
            mark=square, 
            mark options={solid}
        ]
          table[row sep=crcr]{%
          4 0.28333\\
          6 0.22\\
          7 0.11333\\
          8 0.04667\\
          10 \fpeval{18/25}\\
          12 \fpeval{2/3}\\
          15 \fpeval{11/50}\\
          16 \fpeval{119/100}\\
          18 \fpeval{109/100}\\
          22 \fpeval{703/300}\\
          24 \fpeval{43/20}\\
          26 \fpeval{74/25}\\
          28 \fpeval{83/30}\\
          30 \fpeval{191/75}\\
        };
        \addlegendentry{Average Number of Punctures};
        
        \end{axis}
        
        \end{tikzpicture}
    } \quad 
    \subfigure[$R = 1 / 2$]{
        \begin{tikzpicture}
        
        \definecolor{black}{rgb}{0, 0.0, 0}%
        \definecolor{blue}{rgb}{0, 0, 1}%
        \definecolor{red}{rgb}{1, 0, 0}
        \definecolor{green}{rgb}{0, 0.5, 0}
        
        \begin{axis}[
            font=\footnotesize,
            width=0.4\linewidth,
            height=0.3\linewidth,
            scale only axis,
            xmin=0,
            xmax=32,
            xtick={0,5,10,15,...,35},
            xlabel={Length},  
            ymin=-2,
            ymax=23,
            ytick={-5,0,5,10,15,20,25},
            minor tick num=4,
            ylabel near ticks,
            grid=both,
            grid style={
                line width=.1pt, 
                draw=gray!7
            },
            major grid style={
                line width=.2pt,
                draw=gray!50,
            },
            legend style={
                font=\scriptsize, 
                at={(0,1.0)},
                anchor=north west, 
                fill=white,
                legend cell align=left
            },
            legend columns=1
        ]
        
        \addplot [
            color=black, 
            solid, 
            mark=triangle, 
            mark size=3pt,
            mark options = {solid}
        ]
          table[row sep=crcr]{%
          4 2\\
          6 3\\
          7 4\\
          8 4\\
          10 5\\
          12 6\\
          15 8\\
          16 8\\
          18 9\\
          22 11\\
          24 12\\
          26 13\\
          28 14\\
          30 15\\
        };
        \addlegendentry{Worst-Case Covering Radius};
        
        \addplot [
            color=black, 
            solid, 
            mark=o, 
            mark options = {solid}
        ]
          table[row sep=crcr]{%
          4 1.29\\
          6 1.89\\
          7 2.74333\\
          8 2.71\\
          10 3.31333\\
          12 3.91667\\
          15 \fpeval{829/150}\\
          16 \fpeval{553/100}\\
          18 \fpeval{487/75}\\
          22 \fpeval{1171/150}\\
          24 \fpeval{55/6}\\
          26 \fpeval{493/50}\\
          28 \fpeval{547/50}\\
          30 \fpeval{908/75}\\
        };
        \addlegendentry{Average Covering Radius};
        
        \addplot [
            color=black, 
            solid, 
            mark=square, 
            mark options = {solid}
        ]
          table[row sep=crcr]{%
          4 0.34667\\
          6 0.04\\
          7 0\\
          8 \fpeval{47/60}\\
          10 \fpeval{18/25}\\
          12 \fpeval{9/100}\\
          15 \fpeval{343/300}\\
          16 \fpeval{251/150}\\
          18 \fpeval{163/60}\\
          22 \fpeval{202/75}\\
          24 \fpeval{51/10}\\
          26 \fpeval{53/10}\\
          28 \fpeval{496/75}\\
          30 \fpeval{1273/150}\\
        };
        \addlegendentry{Average Number of Punctures};
        
        \end{axis}
        
        \end{tikzpicture}
    }
    \caption{Simulated average covering radius of a $[q - 1, \lfloor (q - 1) R \rfloor]_{q}$ GRS code given by $\GRScover$ with $\GRSdecode = \GS$ and number of punctures until it succeeds for $5 \le q < 32$}
    \label{fig:covering-list}
\end{figure*}
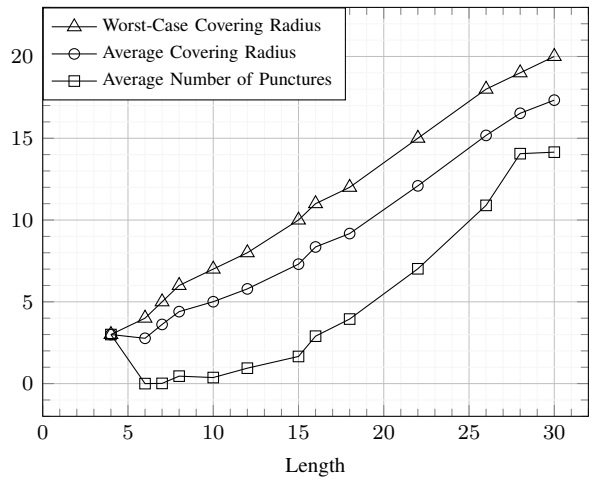
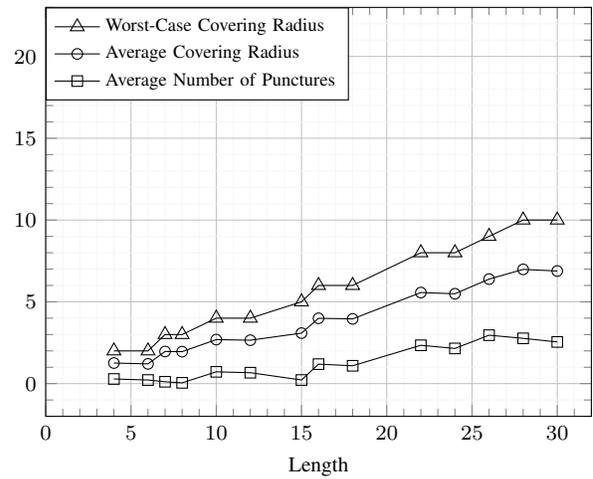
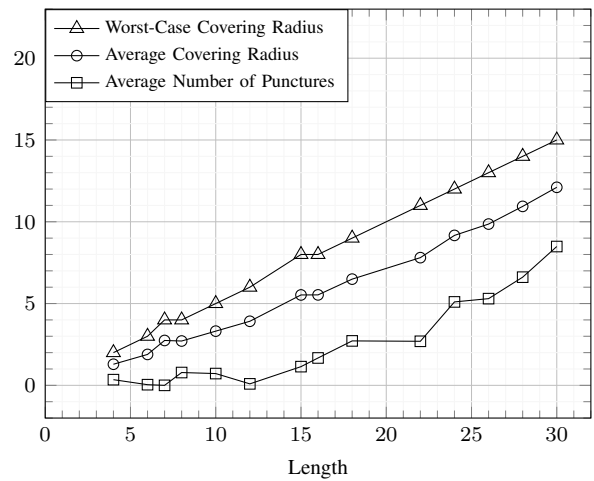

\subsubsection{\texorpdfstring{$\tau_{\GS}$}{tauGS} vs. \texorpdfstring{$\tau_{\max}$}{taumax}}
\label{sec:taumax}

\autoref{fig:coverage} shows the fraction of space covered by Hamming spheres of radius close to $\tau_{\GS}$, and that $\tau_{\GS} = \tau_{\max}$ for $(q, n, k) = (17, 14, 2)$. 
\autoref{fig:comparison} shows a comparison of $\tau_{\GS}$ with $\tau_{\max}$ for $(q, n) = (47, 46)$ and $1 \le k < n$. 
Both support our hypothesis that Hamming spheres of radius $\tau_{\GS}$ cover almost all of the space. 

\begin{figure}[!htbp]
    \centering
    \includegraphics[width=\linewidth]{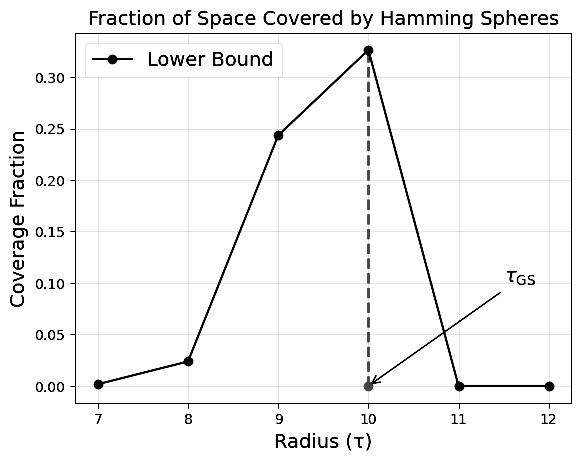}
    \caption{Lower bound on fraction of space covered by Hamming spheres of radius $\tau \in (d / 2, d)$ for a $[14, 2, d]_{17}$ GRS code
    }
    \label{fig:coverage}
\end{figure}

\begin{figure}[!htbp]
    \centering
    \includegraphics[width=\linewidth]{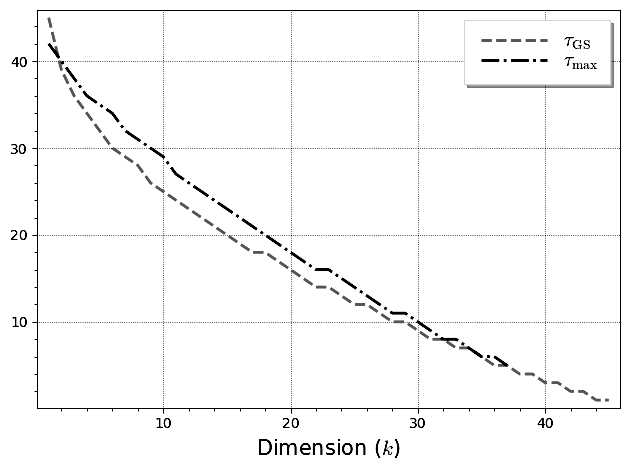}
    \caption{$\tau_{\max}$ vs. $\tau_{\GS}$ for $q = 47, n = 46$ and $1 \le k \le 37$}
    \label{fig:comparison}
\end{figure}

\section{Character-Reed--Solomon Code}
\label{sec:CRS-code}


Recall from \autoref{def:cp} that the message space $\fpkq$ of CP codes consists of polynomials $f \in \CF(k, q)$ with $f_{j p} = 0$ for all integers $j \ge 0$. 
This condition was used to establish minimum distance properties of CP codes~\cite{Hessam22}. Since our focus in this work is on covering rather than packing, we utilize the full message space $\CF(k, q)$ and generalize the definition of CP codes below to what we call \emph{character-Reed--Solomon (CRS) codes}. 

\begin{definition}[CRS Code]
    \label{def:crs}
    Fix $k \le n < q$, a non-trivial character $\chi_{\beta}$ of $\fq$, and distinct evaluation points $\alpha_{1}, \dots, \alpha_{n} \in \fq$. Then, the encoding of $f \in \CF(k, q)$ in the CRS code $\CRS := \CRS_{n, \beta}(\CF(k, q)) \subseteq G_{1, n}(\BC)$ is given by:
    \begin{align}
        \label{eq:crs}
        \CRS(f) 
        &:= (\chi_{\beta}(f(\alpha_{1})), \dots, \chi_{\beta}(f(\alpha_{n}))), 
    \end{align}
    where we identify $\CRS(f)$ with the one-dimensional subspace $\langle \CRS(f) \rangle$. 
\end{definition}

We now derive some theoretical properties of CRS codes, including bounds on the code size. 

\subsection{Theoretical Properties and Bounds on Code Size}

Note that 
$\CRS_{n, {\beta}}(\CF(k, q))$ is essentially the concatenation of a 
RS code with the 
character map $\chi_{\beta}: \fq \to U$, and therefore can be regarded as a RS-coded modulation. 
In particular, $\CRS_{n, \beta}(\CF(k, q))$ prior to concatenation by $\chi_{\beta}$ is simply expressed as $\RS_{n}(\CF(k, q))$ using \eqref{eq:rs}. 
Hence, 
\begin{align}
    \label{eq:crs-size}
    |\CRS| 
    &\le |\RS| = |\CF(k, q)| 
    = q^{k}.
\end{align}
Equality holds if and only if the encoding map \eqref{eq:crs} is injective. 
To see when this is the case, observe that $\CRS(f) = \CRS(g)$ for $f, g \in \CF(k, q)$ if and only if  $ \chi_{\beta}(f(\alpha_{j})) = \chi_{\beta}(g(\alpha_{j}))$ for $j \in \{1, \dots, n\}$, which holds if and only if $\Tr(\beta h(\alpha_{j})) = 0$ for each $j$ by \eqref{eq:character}, 
\hide{\begin{align*}
    & \CRS(f) = \CRS(g) \\
    \iff & \chi_{\beta}(f(\alpha_{j})) 
    = \chi_{\beta}(g(\alpha_{j})) 
    \ \forall j \\ 
    \iff 
    & \exp\left(\frac{2 \pi i \Tr(\beta f(\alpha_{j}))}{p}\right) 
    = \exp\left(\frac{2 \pi i \Tr(\beta g(\alpha_{j}))}{p}\right) 
    \ \forall j \\ 
    \iff 
    & \exp\left(\frac{2 \pi i \Tr(\beta h(\alpha_{j}))}{p}\right) = 1 
    \ \forall j \\ 
    \iff 
    & \Tr(\beta h(\alpha_{j})) = 0 
    \ \forall j,
\end{align*} 
}where $h = f - g \in \CF(k, q)$. 
The following lemma characterizes the zeros of $\Tr$. 

\begin{lemma}
    \label{lem:trace}
    Let $\alpha \in \fq$. Then, $\Tr(\alpha) = 0$ if and only if 
    \begin{align}
        \label{eq:trace}
        \alpha 
        &= \beta^{p} - \beta,
    \end{align}
    for some $\beta \in \fq$. 
    If $\beta = \beta_{0}$ is a solution to \eqref{eq:trace}, then all solutions are given by $\beta_{\lambda} = \beta_{0} + \lambda$ for $\lambda \in \fp$. 
\end{lemma}

\begin{proof}
    Consider the map: 
    \begin{align}
        \label{eq:phi-def}
        \phi: \fq \to \fq, 
        \quad \beta \mapsto \beta^{p} - \beta. 
    \end{align}
    Note that $\phi$ is $\fp$-linear with $\ker(\phi) = \fp$ by Fermat's little theorem. 
    Hence, $|\phi(\fq)| = q / p$ by the first isomorphism theorem. But $\fq \supsetneq \ker(\Tr) \supseteq \phi(\fq)$, whence $\ker(\Tr) = \phi(\fq)$. 
    Thus, if $\alpha \in \ker(\Tr)$, then $\alpha = \phi(\beta_{\lambda})$ for some $\lambda \in \fp$. 
\end{proof}

Therefore, 
$\beta h(\alpha_{j}) =\beta_{j}^{p} - \beta_{j}$ 
for some $\beta_{j} \in \fq$\hide{ by \eqref{eq:trace} and \autoref{lem:trace}}, for each $j$.  Interpolating $u \in {\CF(n, q)}$ over $\{(\alpha_{j}, \beta_{j}): 1 \le j \le n\}$ implies that $u^{p} - u - \beta h$ 
vanishes on the set $\{\alpha_{1}, \dots, \alpha_{n}\}$.
Furthermore, by \autoref{lem:trace}, there are exactly $p^{n}$ such $u$, each corresponding to a realization of $(\beta_{1}, \dots, \beta_{n})$

\hide{(Equivalently, 
we have an isomorphism 
\begin{align*}
    \phi: \fqx / \bigg(\prod_{i = 1}^{n} (X - \alpha_{i})\bigg) &\to \prod_{i = 1}^{n} \fqx / (X - \alpha_{i}) 
    \cong \fqn \\
    P &\mapsto (P(\alpha_{1}), \dots, P(\alpha_{n})) 
\end{align*}
by the Chinese remainder theorem, since the $\alpha_{i}$ are distinct. Hence, there exists a unique polynomial $P \in \CF(n, q)$ such that $\phi(u) = (u(\alpha_{1}), \dots, u(\alpha_{n}))$, where $u: \alpha \mapsto \alpha^{p} - \alpha$.)}

We thus obtain the following result. 

\begin{theorem}
    \label{thm:injective}
    Let $\CC = \CRS_{n, {\beta}}(\CF(k, q))$ with evaluation points $\alpha_{1}, \dots, \alpha_{n}$. 
    Then, $\CC(f) = \CC(g)$ \hide{for $f, g \in \CF(k, q)$ }if and only if 
    \begin{align}
        \label{eq:thm1}
        \beta & (f - g) 
        \equiv u^{p} - u\hide{$ is a multiple of $w$, where $w(X) = }\pmod{(X - \alpha_{1}) \cdots (X - \alpha_{n})}, 
    \end{align}
    for some $u \in \CF(n, q)$. 
    In particular, given the $\fp$-linear map 
    \begin{align}
        \label{eq:T-def} 
        T: \CF(k, q) \to \fp^{n}, \quad 
        f \mapsto (\Tr(\beta f(\alpha_{j}))\hide{, \dots, \Tr(\beta f(\alpha_{n}))})_{j = 1}^{n}, 
    \end{align}
    we have 
    \begin{align}
        \label{eq:CRS-size}
        |\CC| 
        &= |\Image(T)| 
        = p^{\rank(T)}.
    \end{align} 
    Furthermore, 
    if there exists a $u \in \CF(n, q)$ satisfying \eqref{eq:thm1}, then there are exactly $p^{n}$ such $u$.  
    Finally, if \eqref{eq:thm1} is an equality\hide{ $\deg(u) < k / p$ for some $u$}, i.e., 
    \begin{align}
        \label{eq:u}
        \beta (f - g) 
        &= u^{p} - u 
    \end{align}
    for some $u \in \CF(k / p, q)$, 
    then \eqref{eq:u} holds for exactly $p$ such $u$. Otherwise, $\deg(u) \ge n / p$ for all $u$ satisfying \eqref{eq:thm1}. 
\end{theorem}

\begin{proof}
    Only the claims regarding \eqref{eq:u} remain to be proved. 
    
    Since $\deg(f - g) < k < n$, it follows from \eqref{eq:thm1} that $\deg(u) \ge n / p$ or $\deg(u) < k / p$, with \eqref{eq:thm1} taking the form of an equality precisely in the latter case. 

    Consider the map $\phi$ in \eqref{eq:phi-def} extended to $\fqx$. Then $\phi$ is $\fp$-linear with $\ker(\phi) = \fp$. 
    Therefore, if $h \in \phi(\fqx)$, i.e., $h = u^{p} - u$ for some $u \in \fqx$, then there are exactly $p$ such $u \in \fqx$, given by $u + \lambda$ for $\lambda \in \fp$. 
\end{proof}

Observe that \eqref{eq:thm1} cannot be satisfied for $f \neq g$ when $k$ is sufficiently small. This is made precise in the corollary below. 

\begin{corollary}
    Let $\CC = \CRS_{n, {\beta}}(\CF(k, q))$. If $k < 1 + n p / q$, then $|\CC| = q^{k}$. 
\end{corollary}

\begin{proof}
    We claim that, under the given hypothesis, the map $T$ given by \eqref{eq:T-def} is injective. To see this, note that if $T(f) = 0$, then the polynomial $\Tr \circ \beta f \in \fqx$ has zeros at the $n$ evaluation points of $\CC$. But 
    \begin{align*}
        \deg(\Tr \circ \beta f) 
        &= \deg(\Tr) \deg(\beta f) \le (k - 1) q / p < n, 
    \end{align*}
    which forces $\Tr \circ \beta f = 0$, i.e. $f = 0$, thereby proving the claim. 
    Thus, $|\CC| = p^{\dim_{\fp} \CF(k, q)} = q^{k}$ by \eqref{eq:CRS-size}, as desired. 
\end{proof}

In general, \autoref{thm:injective} implies that $|\CC|$ must be equal to a power of $p$ less than or equal to $p^{n}$. Combined with \eqref{eq:crs-size}, this gives the following bounds. 

\begin{theorem} 
    \label{thm:crs-size}
    Let $\CC = \CRS_{n, {\beta}}(\CF(k, q))$.
    Then $|\CC|$ is a power of $p$ such that 
    \begin{align}
        \label{eq:crs-size-2}
        \frac{p^{n}}{q^{n - k}} 
        &\le |\CC| 
        \le \min\{p^{n}, q^{k}\}. 
    \end{align}
    Equality holds on both sides of \eqref{eq:crs-size-2} when $k = n$ or $q = p$. 
\end{theorem}

\begin{proof}
    Let $\alpha_{1}, \dots, \alpha_{n}$ be the evaluation points of $\CC$, 
    and put $w(X) = (X - \alpha_{1}) \cdots (X - \alpha_{n})$. Let $R := \fqx / (w)$ and consider the map $\phi$ in \eqref{eq:phi-def} extended to $R$. By the Chinese remainder theorem, 
    \begin{align}
        \label{eq:crt}
        R 
        &\cong \prod_{i = 1}^{n} 
        \fqx / (X - \alpha_{i}) 
        \cong \fqn, 
    \end{align}
    so that $\ker(\phi) \cong \fp^{n}$ by \autoref{lem:trace}. 
    Moreover, by \autoref{thm:injective}, 
    \begin{align*}
        \CC(f) = \CC(g) 
        &\iff \beta (f - g) \in S \cap \phi(R), 
    \end{align*}
    where $S = \CF(k, q) / (w) \subseteq R$. Therefore, 
    by the isomorphism theorems, 
    \begin{align}
        \label{eq:CRS-size-lb}
        |\CC| 
        &= \frac{|S|}{|S \cap \phi(R)|} 
        \ge \frac{|S|}{|\phi(R)|} 
        = \frac{q^{k}}{(q / p)^{n}}, 
    \end{align} 
    and 
    \begin{align}
        \label{eq:CRS-size-ub}
        |\CC| 
        &= \frac{|S + \phi(R)|}{|\phi(R)|} 
        \le \frac{|R|}{|\phi(R)|} 
        = |\ker(\phi)| 
        = p^{n}. 
    \end{align}
    Thus, \eqref{eq:crs-size-2} follows from \eqref{eq:CRS-size-lb}, \eqref{eq:CRS-size-ub} and \eqref{eq:crs-size}. 
    Equality holds in \eqref{eq:CRS-size-lb} if and only if $\phi(R) \subseteq S$,
    which is satisfied for $k = n$ or $q = p$ (the latter since $u(\alpha_{i})^{q}  = u(\alpha_{i})$ for all $i$). 
    Likewise, equality holds in \eqref{eq:CRS-size-ub} if and only if $S + \phi(R) = R$,
    which is satisfied for sufficiently large $k$, in particular, for $k = n$. 
\end{proof}

Note that if $k < n$, it is possible that neither $\phi(R) \subseteq S$ nor $S + \phi(R) = R$ holds. On the other hand, if $q = p$, then $\phi(R) = \{0\}$. Hence, $S + \phi(R) = R$ holds if and only if $k = n$. 

\subsection{Covering Algorithm}
\label{sec:covering-cp}

As mentioned in the introduction, CP codes have recently been used for the quantization problem in MISO systems with limited feedback~\cite{gooty2025precodingdesignlimitedfeedbackmiso}, 
and the quantization error was characterized using the covering radius of CP codes. 
Moreover, a decoding algorithm for CP codes over prime fields was recently proposed~\cite{riasat2024decodinganalogsubspacecodes}. 
However, as discussed at the beginning of \autoref{sec:covering-grs}, decoding algorithms do not produce good coverage due to the decoding radius being much smaller than the covering radius. Therefore, we propose \autoref{alg:CP-cover} below for covering lines in $G_{1, n}(\BC)$ using CRS codes. 
To describe the idea behind the algorithm, consider the map $\psi_{\beta}$ defined in \eqref{eq:chi-inverse},
which can be regarded as an extension of the $1$-to-$q / p$ map $\chi_ {\beta}^{-1}$ to all of $\BC$, obtained by first mapping any $z \in \BC$ to its closest point (in euclidean distance) in $\chi_{\beta}(\fq)$ then taking its pre-image: 
\begin{align}
    \label{eq:chi-inverse}
    \psi_ {\beta}: \BC &\to \binom{\fq}{q / p} \\ 
    z &\mapsto \chi_ {\beta}^{-1} \left( \exp \left( \frac{2 \pi i}{p} \left \lfloor \frac{p \arg(z)}{2 \pi} + \frac{1}{2} \right \rfloor \right) \right) \notag. 
\end{align}
Here, $\binom{S}{t}$ conventionally denotes the set of all subsets of $S$ of size $t$. 
The given subspace $\langle y\rangle \in G_{1, n}(\BC)$ is first transformed into a $\RS$ codeword by applying $\psi_ {\beta}$ to each $y_{i}$ and picking a random element in $\chi_ {\beta}^{-1}(y_{i})$, which is then followed by a call to $\GRScover$ on the resulting vector. Finally, the message polynomial corresponding to the output is encoded in the CRS code and returned by the algorithm. 

\begin{algorithm}[!htbp]
    \caption{$\CRScover(\CC, y)$} 
    \label{alg:CP-cover}
    \begin{algorithmic}[1]
        \Require{CRS code $\CC = \CRS_{n, {\beta}}(\CF(k, q))$, arbitrary subspace $\langle y\rangle \in G_{1, n}(\BC)$}
        \Ensure{A codeword $c \in \CC$ with $\dc(\langle y\rangle, \langle c\rangle) \le \rhos(\CC)$}
        \BlackBox{$\psi_{\beta}$, $\GRScover$, $\GRSdecode$}
        \State $\CC' \gets \RS_{n}(\CF(k, q))$
            \For {$i = 1, \dots, n$}
            \State \label{step:CRScover:psi} $u_{i} \gets \hbox{random element in } \psi_{\beta}(y_{i})$ 
            \EndFor
            \State $c' \gets \GRScover(\CC', u)$ \label{step:CRScover:cover} 
            \State $f \gets \GRSdecode(\CC', c')$ \label{step:CRScover:decode}  
        \State \Return $\CC(f)$
    \end{algorithmic}
\end{algorithm} 

In a practical implementation of \autoref{alg:CP-cover}, one may avoid \autoref{step:CRScover:decode} altogether by modifying $\GRScover$ in \autoref{alg:GRS-cover} to directly return the message polynomial rather than its re-encoding in \autoref{step:CRScover:cover}. Moreover, the pre-images $u_{j} \in \psi_{\beta}(y_{j})$ in \autoref{step:CRScover:psi} under the character map $\chi_{\beta}$ may be chosen using a ``best-of-$N$'' approach for some fixed $N$. In other words, out of $N$ random assignments of the elements of $\mathbb{F}_{q}$ to the $u_{i}$ under $\psi_{\beta}$, the one giving the smallest chordal distance of the $\CRS$ codeword from $y$ is chosen. 
This is because, in general, each point in $\chi_{\beta}(\mathbb{F}_{q})$ has $q / p$ distinct pre-images in $\mathbb{F}_{q}$, so a brute-force approach would need to test $(q / p)^{n}$ assignments to find the optimal one.

\subsection{Improved 
Upper Bound on Average Covering Radius
}

\label{sec:Covering radius}




Let $\Psi = \chi_{\beta}(\fq)$. Then $\Psi$ consists of the $p$-th roots of unity. Given the input vector $y \in \BC^{n}$, let $\theta_{j} = \arg(y_{j}) \in [0, 2 \pi)$ 
and, 
for $j \in \{1, \dots, n\}$, define 
\begin{align}
    \label{eq:ycheck}
    \check{y}_{j} 
    &:= \exp\left(\frac{2 \pi i}{p} \left \lfloor \frac{p \theta_{j}}{2 \pi} + \frac{1}{2} \right \rfloor\right).
\end{align}
In other words, $\check{y}_{j}$ 
is the closest point to $y_{j}$ 
in $\Psi$ (in Euclidean distance). 
The following theorem improves the upper bound on the average chordal covering radius of 
$\CCP_{n, \beta}(\fpkp) = \CRS_{n, {\beta}}(\CF(k, p))$ by Gooty et al.~\cite[Theorem~1]{gooty2025precodingdesignlimitedfeedbackmiso}.

\begin{theorem}
    \label{thm:gain}
    Consider a distribution $\CD$ on $\BC^{n}$\hide{\hl{$\CS$}}, where the amplitudes \hide{of the coordinates }are i.i.d. with mean $\mu$ and variance $\sigma^{2}$, and are independent of the phases. 
    Then, over $\CD$, 
    the average covering radius\hide{, defined in \autoref{def:mean-covering-radius},} of $\CC := \CRS_{n, \beta}(\CF(k, p))$ \hide{code over $\fp$ of length $n$ and}of rate $R := k / n $ such that 
    \begin{align}
        \label{eq:rate-condition} 
        R 
        &\ge \frac{1}{c + 1} + \frac{(1 - c) \sigma^{2}}{2 (1 + c) n \mu^{2}}, 
    \end{align}
    where $c := \sqrt{\cos(2 \pi / p)}$, is bounded as follows:
    \begin{align}
        \label{eq:mean-covering-radius}
        &\rhosbar(\CC) \le \min\bigg\{\sqrt{1 - \dfrac{(c R + R - 1)^{2} \mu^{2}}{\mu^{2} + \sigma^{2}}}, \\ 
        &\resizebox{0.95\hsize}{!}{$\displaystyle\sqrt{1 - \dfrac{c R + R - 1}{c} \cdot \dfrac{(c R + R - 1)^{2} \mu^{2} + (c^{2} R - R + 1) \sigma^{2} / n}{\mu^{2} + \sigma^{2}}}$}\bigg\}. \notag
    \end{align} 
\end{theorem}

\begin{proof}
    The first expression inside the $\min$ in \eqref{eq:mean-covering-radius} is the square root of the bound on the mean-squared quantization error given by Gooty et al.~\cite[Theorem~1]{gooty2025precodingdesignlimitedfeedbackmiso}. 
    To obtain the improved bound, let 
    $f \in \CF(k, p)$ denote the output of $\GRSdecode$ in \autoref{step:CRScover:decode} of $\CRScover(\CC, y)$. Then, 
    by the triangle inequality, 
    \begin{align}
        |\langle y, \CC(f) \rangle| 
        &\ge 
        \abs{\abs{\sum_{i \in \CA} y_{i}^{*} \chi_{\beta}(f(\alpha_{i}))} - \abs{\sum_{i \in \CB} y_{i}^{*} \chi_{\beta}(f(\alpha_{i}))}} \notag \\ 
        &\ge \abs{\sum_{i \in \CA} y_{i}^{*} \chi_{\beta}(f(\alpha_{i}))} - \sum_{i \in \CB} |y_{i}|, 
        \label{eq:inner-product}
    \end{align}
    where 
    \begin{align*}
        \begin{aligned}
            \CA &= \{1 \le i \le n: \check{y}_{i} = \chi_{\beta}(f(\alpha_{i}))\}, 
            \\ 
            \CB &= \{1 \le i \le n: \check{y}_{i} \neq \chi_{\beta}(f(\alpha_{i}))\}.
        \end{aligned}
    \end{align*}
    To estimate the first term, note that $\check{y}_{j} = e^{i \delta_{j}} y_{j} / {|y_{j}|}$ 
    for $j \in \{1, \dots, n\}$ by \eqref{eq:ycheck}, 
    where 
    \begin{align}
        \label{eq:delta-j}
        \delta_{j} 
        &:= \frac{2 \pi}{p} \left \lfloor \frac{p \theta_{j}}{2 \pi} + \frac{1}{2} \right \rfloor - \theta_{j} 
        \in \left[- \frac{\pi}{p}, \frac{\pi}{p}\right),
    \end{align}
    so that 
    for $j \in \CA$, 
    \begin{align}
        \label{eq:y-delta-j}
        y_{j}^{*} \chi_{\beta}(f(\alpha_{j})) 
        &= y_{j}^{*} \check{y}_{j} 
        = |y_{j}| e^{i \delta_{j}}. 
    \end{align}
    Thus, 
    by \eqref{eq:delta-j} and \eqref{eq:y-delta-j},
    \begin{align*}
        &\hspace{-2em} 
        \abs{\sum_{j \in \CA} y_{j}^{*} \chi_{\beta}(f(\alpha_{j}))}  
        = \abs{\sum_{j \in \CA} |y_{j}| e^{i \delta_{j}}} \\ 
        &= \sqrt{\sum_{j \in \CA} |y_{j}|^{2} + 2 \sum_{\substack{j, k \in \CA \\ j < \ell}} |y_{j}| |y_{\ell}| \cos (\delta_{j} - \delta_{\ell})} \\ 
        &\ge \sqrt{\sum_{j \in \CA} |y_{j}|^{2} + 2 \sum_{\substack{j, \ell \in \CA \\ j < \ell}} |y_{j}| |y_{\ell}| \cos (2 \pi / p)} \\ 
        &\ge 
        c \sum_{j \in \CA} |y_{j}|. 
    \end{align*}
    Let 
    \begin{align*}
        \CE 
        &:= \bigg\{y \in \CD: c \sum_{j \in \CA} |y_{j}| 
        \ge \sum_{j \in \CB} |y_{j}|\bigg\}. 
    \end{align*}
    By Markov's inequality\hide{ and  
    \autoref{lem:rho-cp}}, 
    \begin{align}
        \Pr[\CE]
        &= \Pr\bigg[c \sum_{j = 1}^{n} |y_{j}| \ge (c + 1) \sum_{j \in \CB} |y_{j}|\bigg] \notag \\ 
        &\ge 1 - \frac{\expec\left[\sum_{j \in \CB} |y_{j}| / \sum_{j = 1}^{n} |y_{j}|\right]}{c / (c + 1)}. 
        \label{eq:markov} 
    \end{align}
    Letting 
    $A := |\CA|$, $B := |\CB|$,  
    $Y := \sum_{j = 1}^{n} |y_{j}|$, 
    and using the fact that the $|y_{j}|$ are i.i.d., we have \hide{identical distribution of fact that the $|y_{j}|$ are i.i.d. 
    we have, by the law of total expectation,} 
    \begin{align}
        \expec\left[\frac{\sum_{j \in \CB} |y_{j}|}{Y}\right] 
        &= \expec\left[\sum_{j = 1}^{n} \frac{|y_{j}|}{Y} {\bf 1}_{\{j \in \CB\}}\right] \notag \\ 
        &= \sum_{j = 1}^{n} \expec\left[\frac{|y_{j}|}{Y}\right] \expec[{\bf 1}_{\{j \in \CB\}}] \notag \\ 
        &= \frac{\expec[B]}{n} \notag \\ 
        &\le \frac{n - k}{n} \notag \\ 
        &= 1 - R, \label{eq:expec} 
    \end{align}
    which follows by linearity and independence of the amplitudes and phases. 
    Here we have used 
    \begin{align}
        \label{eq:covering-radius}
        B = n - A \le \rhoh(\CC) = \rhoh(\RS_{n}(\CF(k, p)) = 
        n - k, 
    \end{align} 
    which follows from the definition of $\rhoh$ (see \eqref{eq:rho} and its associated remarks) and the injectivity of $\chi_{\beta}: \fp \to \BC$. 
    Thus, 
    by the law of total expectation, 
    \begin{align}
        &
        \expec[|\langle y, \CC(f)\rangle|^{2}] 
        \ge \expec[|\langle y, \CC(f)\rangle|^{2} | \CE] \Pr[\CE] 
        \notag \\ 
        &\hspace{-1mm}\ge \left(1 - \frac{1 - R}{c / (c + 1)}\right) \cdot \expec\left[\bigg(c \sum_{i \in \CA} |y_{i}| - \sum_{i \in \CB} |y_{i}|\bigg)^{2} \bigg| \CE\right] \notag \\ 
        &\hspace{-1mm}= \frac{c R + R - 1}{c} \cdot \expec\left[\expec\left[\bigg(c \sum_{i \in \CA} |y_{i}| - \sum_{i \in \CB} |y_{i}|\bigg)^{2} \bigg| \CA\right] \bigg| \CE\right] \notag \\ 
        &\hspace{-1mm}= \frac{c R + R - 1}{c} \cdot \expec[(c A - B)^{2} \mu^{2} + (c^{2} A + B) \sigma^{2} | \CE] 
        \notag \\ 
        &\hspace{-1mm}\ge 
        \frac{c R + R - 1}{c} \cdot [(c k + k - n)^{2} \mu^{2} + (c^{2} k + n - k) \sigma^{2}],
        \label{eq:exp-inner-product-squared} 
    \end{align}
    which follows from \eqref{eq:inner-product}, \eqref{eq:markov} and \eqref{eq:expec} upon noting that 
    \begin{itemize}
        \item $(c A - B)^{2} \mu^{2} + (c^{2} A + B) \sigma^{2}$ is increasing in $A$ for 
        \begin{align}
            \label{eq:acondition}
            A 
            &\ge \frac{n}{1 + c} + \frac{(1 - c) \sigma^{2}}{2 (1 + c) \mu^{2}}. 
        \end{align}
        \hide{Assuming the $|y_{j}|$ are i.i.d.,\hide{ with mean $\mu$ and variance $\sigma^{2}$} the probability that the right-hand side of \eqref{eq:ABbound} is negative is 
        \begin{align}
            \hide{\Pr[{\rm RHS} < 0] 
            &= }\Pr\left[\frac{\sum_{j \in B} |y_{j}|}{\sum_{j \in A} |y_{j}|} > c\right]  
            \le \frac{B}{c A} 
            \le \frac{n - k}{c k} = \frac{1 - R}{c R}
            \label{eq:markov}
        \end{align}
        by Markov's inequality, 
        Therefore, \hl{as $R\to 1$},\hide{ which tends to $0$ as $R \to 1$, i.e. ${\rm RHS} \ge 0$ almost surely$\expec[{\rm RHS}] = (c A - B) \mu \ge (c k + k - n) \mu > 0$ so the RHS is positive on average! In particular, (This is only true if the RHS of \eqref{eq:ABbound} is positive! Q: If $X \ge Y$ and $\expec[Y] \ge 0$, is $\expec[X^{2}] \ge \expec[Y^{2}]$?
        
        Note: iff $\var(X) + \expec[X]^{2} \ge \var(Y) + \expec[Y]^{2}$. In the worst case, we cannot say anything more than this since $y$ can be perpendicular to $\CRS(f)$. Nevertheless, $\expec[{\rm RHS}] = (c a - B) \mu \ge (c k + k - n) \mu > 0$ so the RHS is positive on average! In particular, by Markov's inequality, (see \href{https://math.stackexchange.com/questions/1258485/expectation-of-quotient-of-random-variables}{here})
        \begin{align*}
            \Pr[{\rm RHS} < 0] 
            &= \Pr\left[\frac{\sum_{j \in B} |y_{j}|}{\sum_{j \in A} |y_{j}|} > c\right]  
            \le \frac{B}{c A} 
            \le 
            \frac{1 - R}{c}
        \end{align*}
        which tends to $0$ as $R \to 1$, i.e. ${\rm RHS} \ge 0$ almost surely!)} 
        \begin{align}
            \expec
            [\langle y,  \CRS(f) \rangle^{2}] 
            &\ge \expec
            \bigg[\bigg(\sqrt{\cos(2 \pi / p)} 
            \sum_{j \in A} |y_{j}| - \sum_{j \in B} |y_{j}|\bigg)^{2}\bigg] \notag \\ 
            &= (cA - n + A)^{2} \mu^{2} + (c^{2} A + n - A) \sigma^{2}
            \label{eq:g(a)}
        \end{align}
        \hl{almost surely}. 
        \eqref{eq:g(a)} is \hide{strictly convex and }increasing in $A$ for $A \ge n R = k$ 
        and 
        \begin{align}
            R
            &\ge \frac{1}{c + 1} + \frac{(1 - c) \sigma^{2}}{2 (c + 1) n \mu^{2}} 
            . 
            \label{eq:rate-bound}
        \end{align}
        Assuming \eqref{eq:nbound} holds, \eqref{eq:rate-bound} is satisfied for $R \to 1$. 
        Therefore, 
        \begin{align*}
            \expec
            [\langle y, \CRS(f) \rangle^{2}]
            &\ge 
            (c k + k - n)^{2} \mu^{2} + (c^{2} k + n - k) \sigma^{2}
            ,
        \end{align*}
        by the Cauchy--Schwarz inequality and \eqref{eq:exp-inner-product}. and that}
        \item \eqref{eq:acondition} is implied by \eqref{eq:rate-condition}, since $A \ge k = n R$ 
        \hide{almost surely }by \eqref{eq:covering-radius}\hide{, so that \eqref{eq:acondition} holds by \eqref{eq:abound}}.
    \end{itemize} 
    Applying Jensen's inequality twice
    followed by \eqref{eq:exp-inner-product-squared} gives: 
    \begin{align}
        &\rhosbar(\CC)
        = \expec\left[\sqrt{1 - \frac{|\langle y, \CC(f) \rangle|^{2}}{n \norm{y}^{2}}}\right] \notag \\ 
        &\le \sqrt{1 - \expec\bigg[\frac{|\langle y, \CC(f) \rangle|^{2}}{n \norm{y}^{2}}\bigg]} \notag \\ 
        &\le 
        \sqrt{1 - \frac{c R + R - 1}{c} \cdot \frac{(c k + k - n)^{2} \mu^{2} + (c^{2} k + n - k) \sigma^{2}}{n \expec[\norm{y}^{2}]}}. 
        \notag
    \end{align}
    Substituting $\expec[\norm{y}^{2}] = n (\mu^{2} + \sigma^{2})$ yields the conclusion. 
\end{proof}

\autoref{thm:gain} provides sufficient SNR and rate conditions under which CRS codes achieve a small average covering radius. Comparing this bound with the random-coding expectation of \autoref{thm:random-chordal} yields the following corollary, which shows that CRS codes achieve average covering radii within a constant factor of random codes 
at high rate and high SNR.

\begin{corollary}
    \label{cor:random-coding-bound}
    Consider a distribution $\CD$ on $\BC^{n}$, where the amplitudes are i.i.d. with mean $\mu$ and variance $\sigma^{2}\hide{ < \mu^{2} / (p - 1)}$, and independent of the phases, 
    such that $\mu^{2} / \sigma^{2} \to \infty$ as $n \to \infty$. 
    Then,
    for sufficiently large $n$, 
    the average covering radius over $\CD$ of a $\CRS$ code of rate $R > 1 - 1 / (\ln p)$ over $\fp$ 
    is within a factor of $2 \sqrt{p (1 - R)}$ 
    of the average covering radius of a random code of size $p^{n R}$ in $G_{1, n}(\BC)$. 
\end{corollary}

\begin{proof}
    By \autoref{thm:random-chordal}, the expected average covering radius of a random code of size $M = p^{n R}$ is equal to 
    \begin{align}
        \expec_{\CC_{\CY}} [\rhosbar(\CC_{\CY})] 
        &= \frac{M (M - 1)}{a_{n} - 1} \betaf(a_{n}, M - 1) \notag \\ 
        &= \frac{M (M - 1) \Gamma(a_{n}) \Gamma(M - 1)}{(a_{n} - 1) \Gamma(a_{n} + M - 1)} \notag \\ 
        &= \frac{\Gamma(a_{n} - 1) \Gamma(M + 1)}{\Gamma(a_{n} + M - 1)}, 
        \label{eq:random-new}
    \end{align} 
    where and $\Gamma$ is the gamma function and 
    \begin{align}
        \label{eq:an-def}
        a_{n} 
        &:= 2 + \frac{1}{2 (n - 1)}. 
    \end{align}
    Using Stirling's approximation, 
    as $n \to \infty$, 
    \begin{align}
        \frac{\Gamma(M + 1)}{\Gamma(a_{n} + M - 1)} 
        &\sim \frac{M^{M + 1 / 2} e^{-M}}{(a_{n} + M - 2)^{a_{n} + M - 3 / 2} e^{-(a_{n} + M - 2)}} 
        \notag \\ 
        &= e^{a_{n} - 2} \frac{M^{M + 1 / 2}}{(a_{n} + M - 2)^{a_{n} + M - 3 / 2}} 
        \notag \\ 
        &= e^{a_{n} - 2} \frac{M^{M + 1 / 2 - a_{n} - M + 3 / 2}}{(\frac{a_{n} + M - 2}{M})^{a_{n} + M - 3 / 2}} 
        \notag \\ 
        &= \left(\frac{e}{M}\right)^{a_{n} - 2} \frac{1}{(1 + \frac{a_{n} - 2}{M})^{a_{n} + M - 3 / 2}} 
        \notag \\ 
        &\sim \left(\frac{e}{M}\right)^{a_{n} - 2} \frac{1}{e^{a_{n} - 2} (1 + \frac{a_{n} - 2}{M})^{a_{n} - 5 / 2}} 
        \notag \\ 
        &\sim M^{2 - a_{n}} 
        \notag \\ 
        &= M^{-\frac{1}{2 (n - 1)}} 
        \notag \\ 
        &\sim p^{-\frac{R}{2}}, 
        \label{eq:limit}
    \end{align}
    \hide{Using Stirling's approximation, 
    \begin{align*}
        &\frac{\Gamma(M + 1)}{\Gamma(a_{n} + M - 1)} \\ 
        &= \frac{\sqrt{2 \pi} M^{M + 1 / 2} e^{-M} (1 + O(\frac{1}{M}))}{\sqrt{2 \pi} (a_{n} + M - 2)^{a_{n} + M - 3 / 2} e^{-(a_{n} + M - 2)} (1 + O(\frac{1}{M}))} \\ 
        &= e^{a_{n} - 2} \left(1 + O\bigg(\frac{1}{M}\bigg)\right) \frac{M^{M + 1 / 2}}{(a_{n} + M - 2)^{a_{n} + M - 3 / 2}} \\ 
        &= e^{a_{n} - 2} \left(1 + O\bigg(\frac{1}{M}\bigg)\right) \frac{M^{M + 1 / 2 - a_{n} - M + 3 / 2}}{(\frac{a_{n} + M - 2}{M})^{a_{n} + M - 3 / 2}} \\ 
        &= \left(\frac{e}{M}\right)^{a_{n} - 2} \left(1 + O\bigg(\frac{1}{M}\bigg)\right) \frac{1}{(1 + \frac{a_{n} - 2}{M})^{a_{n} + M - 3 / 2}} \\ 
        &> \left(\frac{e}{M}\right)^{a_{n} - 2} \left(1 + O\bigg(\frac{1}{M}\bigg)\right) \frac{1}{e^{a_{n} - 2} (1 + \frac{a_{n} - 2}{M})^{a_{n} - 5 / 2}} \\ 
        &= \left(\frac{1}{M}\right)^{a_{n} - 2} \left(1 + O\bigg(\frac{1}{M}\bigg)\right) \frac{1}{(1 + \frac{a_{n} - 2}{M})^{a_{n} - 5 / 2}} \\ 
        &= \left(\frac{1}{M}\right)^{a_{n} - 2} \left(1 + O\bigg(\frac{1}{M}\bigg)\right) \left(\frac{M}{a_{n} + M - 2}\right)^{a_{n} - 5 / 2} \\ 
        &= M^{- 1 / 2} (a_{n} + M - 2)^{5 / 2 - a_{n}} \left(1 + O\bigg(\frac{1}{M}\bigg)\right) \\ 
        &> M^{2 - a_{n}} \left(1 + O\bigg(\frac{1}{M}\bigg)\right) \\ 
        &= M^{-\frac{1}{2 (n - 1)}} \left(1 + O\bigg(\frac{1}{M}\bigg)\right) \\ 
        &> p^{-\frac{R}{2}} \left(1 + O\bigg(\frac{1}{M}\bigg)\right) 
    \end{align*}
    By Gautschi's inequality~\cite{Gautschi1959SomeEI}, 
    \begin{align}
        \label{eq:wendel}
        \frac{\Gamma(M + 1)}{\Gamma(a_{n} + M - 1)} 
        &\ge \frac{M^{3 - a_{n}}}{a_{n} + M - 2} 
        = \frac{p^{- \frac{R}{2} (1 + \frac{1}{n - 1})}}{1 + \frac{p^{- n R}}{2 (n - 1)}}
        , 
    \end{align}
    since $a_{n} - 2 \in (0, 1)$ by \eqref{eq:an-def}. 
    As $n \to \infty$, 
    \begin{align}
        \label{eq:limit}
        \frac{p^{- R (1 + \frac{1}{n - 1})} \Gamma(1 + \tfrac{1}{2 (n - 1)})^{2}}{(1 + \frac{1}{2 (n - 1) p^{n R}})^{2}} 
        &\sim p^{-R}, 
    \end{align}
    }and, by \autoref{thm:gain}, 
    \begin{align}
        \rhosbar(\CC) 
        &\le \sqrt{1 - \frac{(c R + R - 1)^{2} \mu^{2}}{\mu^{2} + \sigma^{2}}} 
        \notag \\ 
        &\sim \sqrt{1 - (2 R - 1)^{2}}, 
        \label{eq:crs-bound} 
    \end{align}
    where $\CC = \CRS_{n}(\fpkp)$. 
    Here, the asymptotic notation $A(n) \sim B(n)$ conventionally means that $A(n) / B(n) \to 1$ as $n \to \infty$. 
    Letting $R = 1 - \varepsilon$ for $\varepsilon < 1 / (\ln p)$, we have 
    \begin{align}
        \label{eq:rhosbar-rand} 
        p^{- \frac{R}{2}} 
        &= \sqrt{p^{-1} e^{\varepsilon \ln p}} 
        \sim \sqrt{\frac{1 + \varepsilon \ln p}{p}},
    \end{align}
    and 
    \begin{align}
        \label{eq:rhosbar-crs} 
        \sqrt{1 - (2 R - 1)^{2}}
        &= \sqrt{1 - (1 - 2 \varepsilon)^{2}} 
        \sim 2 \sqrt{\varepsilon}, 
    \end{align}
    i.e., the ratio of the average covering radii of a random code and $\CC$, by \cref{eq:random-new,eq:an-def,eq:limit,eq:crs-bound,eq:rhosbar-rand,eq:rhosbar-crs}, is asymptotically at least 
    \begin{align}
        \label{eq:asymptotic-ratio}
        \frac{p^{- \frac{R}{2}}}{\sqrt{1 - (2 R - 1)^{2}}} 
        &\sim \frac{1}{2 \sqrt{\varepsilon p}} 
        . 
    \end{align}
    Therefore, $\CC$ achieves the average covering radius of random codes within a factor of $2 \sqrt{\varepsilon p}$ \hide{by \eqref{eq:asymptotic-ratio} 
    }for sufficiently large $n$. 
\end{proof}

Note that $2 \sqrt{p (1 - R)} \sim 2 \sqrt{n - k}\hide{ = 2 \sqrt{\rhoh(\CC)}}$ whenever $p = n + o(n)$. 
Classical results from analytic number theory (see, e.g.,~\cite{BHP01}) guarantee the existence of such $p$ for any $n$. 

\subsection{Simulation Results and Comparison with Random Coding Bound}
\label{sec:simulation-crs}

\autoref{fig:chordal-covering-radius} shows the Monte Carlo simulated average \hide{chordal }covering radius of \hide{the $(6, k)_{7}$ CP code }$\CC = \CRS_{6, 1}(\CF(k, 7))$ for $1 \le k < 6$ computed using \autoref{alg:CP-cover} with $\GRSdecode \in \{\hbox{BW}, \hbox{GS}\}$, together with theoretical bounds 
(compare this with \autoref{fig:covering-radius}). 
The input vector $y \in \BC^{n}$ for each simulation was chosen to have i.i.d. coordinates drawn from a complex normal distribution so that the amplitudes follow a Rayleigh distribution with mean $1$\hide{$y_{i} \in \BC$ with $\arg(y_{i}) \in [0, 2 \pi)$ chosen uniformly randomly 
for each $i$}. 
The average chordal distance from $y$ to the output codeword was computed over $500$ data points for each $k$. 
The random coding bound is given by \eqref{eq:random-chordal} in \autoref{thm:random-chordal}. 
The Gooty et al. upper bound is given by the square root of their bound on the mean squared quantization error~\cite[Theorem~1]{gooty2025precodingdesignlimitedfeedbackmiso}, which is improved upon by \autoref{thm:gain}. 

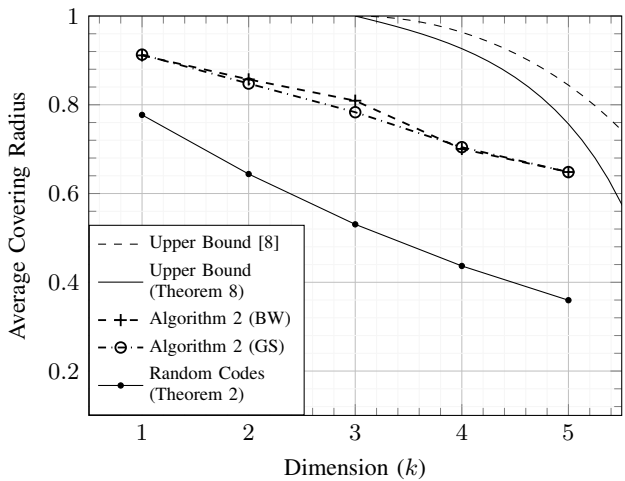
\begin{figure}[!htbp]
    \centering
        \begin{tikzpicture}
        
        \definecolor{black}{rgb}{0, 0.0, 0}%
        \definecolor{blue}{rgb}{0, 0, 1}%
        \definecolor{red}{rgb}{1, 0, 0}
        \definecolor{green}{rgb}{0, 0.5, 0}
        
        \begin{axis}[
            font=\small,
            width=0.8\linewidth,
            height=0.6\linewidth,
            scale only axis,
            xmin=0.5,
            xmax=5.5,
            xtick={0,1,2,3,...,6},
            xlabel={Dimension ($k$)},  
            ymin=0.1,
            ymax=1,
            ylabel={Average Covering Radius},
            ylabel near ticks,
            minor tick num=4,
            grid=both,
            grid style={
                line width=.1pt, 
                draw=gray!7
            },
            major grid style={
                line width=.2pt,
                draw=gray!50,
            },
            legend style={
                font=\scriptsize, 
                at={(0,0)},
                anchor=south west, 
                draw=black,
                fill=white,
                legend cell align=left, 
                cells={align=left},
                row sep=0.1pt
            }
        ]
        
        \addplot[
            color=black, 
            solid, 
            dashed, 
            domain=3:6
        ] {sqrt(1 - (x/6*(1+sqrt(cos(2*pi/7)))-1)^2 * 1/(1 + (8/pi-2))};
        \addlegendentry{Upper Bound~\cite{gooty2025precodingdesignlimitedfeedbackmiso}
        };
        
        \addplot[
            color=black, 
            solid,
            solid, 
            domain=3:6
        ] {sqrt(1 - (x/6*(1+sqrt(cos(2*pi/7)))-1)/sqrt(cos(2*pi/7))*((x/6*(1+sqrt(cos(2*pi/7)))-1)^2*1 + (x/6*(sqrt(cos(2*pi/7))^2-1)+1)^2*(8/pi-2))/(1+(8/pi-2))};
        \addlegendentry{Upper Bound \\
        (\autoref{thm:gain})
        };
        
        \addplot[
            color=black, 
            solid, 
            dashed, 
            line width=0.7pt,
            mark=+, 
            mark size=2.3pt,
            mark options={solid}
        ] table[row sep=crcr]{%
            1 0.911234742664266 \\
            2 0.857170733964536 \\
            3 0.809156953633996 \\
            4 0.700476494648494 \\
            5 0.648471972207539 \\
        };
        \addlegendentry{\autoref{alg:CP-cover} (BW)};
        
        \addplot[
            color=black, 
            solid, 
            dash dot, 
            mark=o, 
            line width=0.7pt,
            mark options={solid}
        ] table[row sep=crcr]{%
            1 0.912941213038983 \\
            2 0.847771926293615 \\
            3 0.783229628604138 \\
            4 0.704410619850270 \\
            5 0.648308254845208 \\
        };
        \addlegendentry{\autoref{alg:CP-cover} (GS)};
        
        \addplot[
            color=black, 
            solid, 
            solid, 
            mark=*, 
            mark size=1pt,
            mark options={solid}
        ] table[row sep=crcr]{%
            1 0.777167527593263 \\
            2 0.643925432693930 \\
            3 0.530568745885669 \\
            4 0.436808931976995 \\
            5 0.359575615705036 \\
        };
        \addlegendentry{Random Codes 
        \\
        (\autoref{thm:random-chordal})};
        
        \end{axis}
        
        \end{tikzpicture}
    \caption{Comparison of simulated average covering radii of a $(6, k)_{7}$ CRS code and random codes for $1 \le k < 6$}
    \label{fig:chordal-covering-radius}
\end{figure}

The behavior observed in \autoref{fig:chordal-covering-radius} highlights an important distinction between covering in the Hamming and Grassmann spaces. While Reed–Solomon-based constructions can outperform random codes in the Hamming space under the average covering radius metric, the corresponding CRS constructions fall short of random codebooks in the Grassmann space. This performance gap can be attributed to a structural constraint inherent to CRS codes: all coordinates of the basis vector have equal amplitude, as defined in \eqref{eq:crs}. Although this constant-modulus property is highly desirable in wireless communication applications, where it enables power-efficient transmission and simplifies hardware implementation, it restricts the geometric flexibility of the codebook in the Grassmann space. In contrast, random codebooks implicitly exploit variations in coordinate amplitudes, which allows for more effective space covering. We believe that achieving or surpassing the performance of non-asymptotic random codes in the Grassmann space requires the design of structured codes with non-uniform amplitude profiles across coordinates. The exploration of such amplitude-adaptive constructions, while preserving practical constraints, is left for future work.

\autoref{fig:asymptotic-bounds-comparison} compares the upper bound on the average covering radius of an $(n,k)_p$ CRS code and the random coding bound in the extreme high-rate regime $k/n \to 1$, where $k = n - 1$ and $p$ is the smallest prime $> n$, for $n \le 10^4$\hide{70571[31, 48523]}. The bounds are evaluated under a source distribution with mean $\mu = n$ and variance $\sigma^2 = 1$. As $n \to \infty$, both the random-coding bound and the upper bound derived in \autoref{thm:gain} vanish, while the ratio between them approaches $2$. In particular, the average covering radius of the CRS code remains sandwiched between the two bounds and within a fixed constant factor of the random-coding bound, in accordance with \autoref{cor:random-coding-bound}.

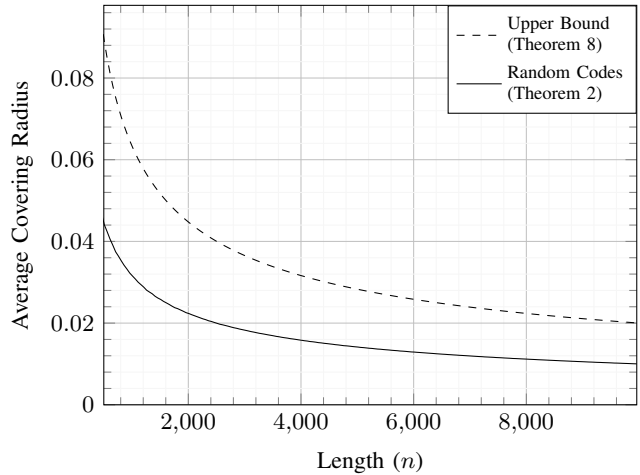
\begin{figure}[!htbp]
    \centering
    \begin{tikzpicture}
        \definecolor{black}{rgb}{0,0,0}%
        \definecolor{blue}{rgb}{0,0,1}%
        \definecolor{red}{rgb}{1,0,0}
        \definecolor{green}{rgb}{0,0.5,0}
        
        \begin{axis}[
           font=\small,
           width=0.8\linewidth,
           height=0.6\linewidth,
           scale only axis,
           xmin=500,
           xmax=9960,
           xlabel={Length ($n$)},
           ymin=0,
           ytick={0,0.02,0.04,0.06,0.08,0.1},
           yticklabel style={
              /pgf/number format/fixed,
           },
           ylabel={Average Covering Radius},
           ylabel near ticks,
           minor tick num=4,
           scaled ticks=false,
           grid=both,
           grid style={
              line width=.1pt,
              draw=gray!7
           },
           major grid style={
              line width=.2pt,
              draw=gray!50,
           },
           legend style={
              font=\scriptsize,
              at={(1,1)},
              anchor=north east,
              draw=black,
              fill=white,
              legend cell align=left,
              cells={align=left},
              row sep=0.1pt
           }
        ]
        
        \addplot[
           color=black,
           solid,
           dashed,
        ] table[row sep=crcr]{%
          10 0.692019867522471 \\ 
          60 0.266179637858962 \\ 
          110 0.193948874447777 \\ 
          160 0.160033715805730 \\ 
          210 0.139341931234788 \\ 
          260 0.124988336483352 \\ 
          310 0.114341499965425 \\ 
          360 0.105986827856851 \\ 
          410 0.0992461789361650 \\ 
          460 0.0936680731611766 \\ 
          510 0.0889036828710934 \\ 
          560 0.0848245771987251 \\ 
          610 0.0812500375736749 \\ 
          660 0.0780940298779198 \\ 
          710 0.0752719211727588 \\ 
          760 0.0727454808310105 \\ 
          810 0.0704527634419962 \\ 
          860 0.0683630859823506 \\ 
          910 0.0664506005133665 \\ 
          960 0.0646873277116655 \\ 
          1010 0.0630605258434957 \\ 
          1060 0.0615499522640577 \\ 
          1110 0.0601410312724033 \\ 
          1160 0.0588268638432212 \\ 
          1210 0.0575944101807857 \\ 
          1260 0.0564339190983036 \\ 
          1310 0.0553444226622596 \\ 
          1360 0.0543156134271143 \\ 
          1410 0.0533394448387585 \\ 
          1460 0.0524158458501399 \\ 
          1510 0.0515394825054252 \\ 
          1560 0.0507039654501636 \\ 
          1610 0.0499087810579788 \\ 
          1660 0.0491495022217364 \\ 
          1710 0.0484232107523848 \\ 
          1760 0.0477283821787274 \\ 
          1810 0.0470641789244158 \\ 
          1860 0.0464258566681775 \\ 
          1910 0.0458126980673381 \\ 
          1960 0.0452227586052389 \\ 
          2010 0.0446562468098860 \\ 
          2060 0.0441097655558204 \\ 
          2110 0.0435830645718041 \\ 
          2160 0.0430746962964300 \\ 
          2210 0.0425836264721853 \\ 
          2260 0.0421089009844776 \\ 
          2310 0.0416500945894450 \\ 
          2360 0.0412053775159557 \\ 
          2410 0.0407753080943323 \\ 
          2460 0.0403578942792005 \\ 
          2510 0.0399531128254964 \\ 
          2560 0.0395601601030687 \\ 
          2610 0.0391791437652754 \\ 
          2660 0.0388087042954966 \\ 
          2710 0.0384485224124500 \\ 
          2760 0.0380979883117916 \\ 
          2810 0.0377569690072292 \\ 
          2860 0.0374251715938938 \\ 
          2910 0.0371016637817681 \\ 
          2960 0.0367866127504593 \\ 
          3010 0.0364794115188238 \\ 
          3060 0.0361797408681733 \\ 
          3110 0.0358871906309588 \\ 
          3160 0.0356018731429878 \\ 
          3210 0.0353230831991089 \\ 
          3260 0.0350507439272414 \\ 
          3310 0.0347847963368871 \\ 
          3360 0.0345247186052620 \\ 
          3410 0.0342703297470544 \\ 
          3460 0.0340215391931429 \\ 
          3510 0.0337780617657312 \\ 
          3560 0.0335396083039947 \\ 
          3610 0.0333063634251307 \\ 
          3660 0.0330777218041913 \\ 
          3710 0.0328538445314707 \\ 
          3760 0.0326345197073380 \\ 
          3810 0.0324193300215246 \\ 
          3860 0.0322085352407679 \\ 
          3910 0.0320017359361852 \\ 
          3960 0.0317987892786455 \\ 
          4010 0.0315997533160797 \\ 
          4060 0.0314042769461689 \\ 
          4110 0.0312125852815043 \\ 
          4160 0.0310241131313626 \\ 
          4210 0.0308392902169431 \\ 
          4260 0.0306575961321240 \\ 
          4310 0.0304789475563264 \\ 
          4360 0.0303036230402132 \\ 
          4410 0.0301311203245904 \\ 
          4460 0.0299616498645418 \\ 
          4510 0.0297949483607636 \\ 
          4560 0.0296310127576617 \\ 
          4610 0.0294696716854204 \\ 
          4660 0.0293110583804552 \\ 
          4710 0.0291548746345834 \\ 
          4760 0.0290011376503040 \\ 
          4810 0.0288500036611225 \\ 
          4860 0.0287010988121837 \\ 
          4910 0.0285544176095988 \\ 
          4960 0.0284100201968908 \\ 
          5010 0.0282678134149281 \\ 
          5060 0.0281276017123954 \\ 
          5110 0.0279896151569351 \\ 
          5160 0.0278535459244795 \\ 
          5210 0.0277194125334113 \\ 
          5260 0.0275873264292355 \\ 
          5310 0.0274569755066851 \\ 
          5360 0.0273284749599432 \\ 
          5410 0.0272018807196405 \\ 
          5460 0.0270769124087011 \\ 
          5510 0.0269536949789721 \\ 
          5560 0.0268321614503980 \\ 
          5610 0.0267121900822560 \\ 
          5660 0.0265938714249490 \\ 
          5710 0.0264771269332470 \\ 
          5760 0.0263618046302003 \\ 
          5810 0.0262481069755630 \\ 
          5860 0.0261358146463156 \\ 
          5910 0.0260248999129233 \\ 
          5960 0.0259154005915907 \\ 
          6010 0.0258073703293583 \\ 
          6060 0.0257005894799483 \\ 
          6110 0.0255951572598977 \\ 
          6160 0.0254909986036966 \\ 
          6210 0.0253881077766595 \\ 
          6260 0.0252864400736159 \\ 
          6310 0.0251859966241199 \\ 
          6360 0.0250867344882665 \\ 
          6410 0.0249886069056698 \\ 
          6460 0.0248916575024942 \\ 
          6510 0.0247958162536074 \\ 
          6560 0.0247011018888319 \\ 
          6610 0.0246074255670122 \\ 
          6660 0.0245148452906028 \\ 
          6710 0.0244232591423793 \\ 
          6760 0.0243327340808515 \\ 
          6810 0.0242431565053904 \\ 
          6860 0.0241546169296383 \\ 
          6910 0.0240670202404881 \\ 
          6960 0.0239803647502412 \\ 
          7010 0.0238946338112752 \\ 
          7060 0.0238098065012576 \\ 
          7110 0.0237258856823964 \\ 
          7160 0.0236428370509012 \\ 
          7210 0.0235607035798486 \\ 
          7260 0.0234793366104199 \\ 
          7310 0.0233988804938701 \\ 
          7360 0.0233192244277445 \\ 
          7410 0.0232403888179172 \\ 
          7460 0.0231622980846056 \\ 
          7510 0.0230850419651931 \\ 
          7560 0.0230085456419291 \\ 
          7610 0.0229327734369048 \\ 
          7660 0.0228577681757964 \\ 
          7710 0.0227834940346067 \\ 
          7760 0.0227098948803525 \\ 
          7810 0.0226370848704839 \\ 
          7860 0.0225649308555643 \\ 
          7910 0.0224934588921632 \\ 
          7960 0.0224226757601621 \\ 
          8010 0.0223525496972002 \\ 
          8060 0.0222830605681919 \\ 
          8110 0.0222142422729843 \\ 
          8160 0.0221460444160919 \\ 
          8210 0.0220784579336678 \\ 
          8260 0.0220115088696176 \\ 
          8310 0.0219451588051706 \\ 
          8360 0.0218793989425499 \\ 
          8410 0.0218142206883277 \\ 
          8460 0.0217496425209984 \\ 
          8510 0.0216856196228447 \\ 
          8560 0.0216221617371672 \\ 
          8610 0.0215592433945209 \\ 
          8660 0.0214968994841662 \\ 
          8710 0.0214350792586626 \\ 
          8760 0.0213737920711619 \\ 
          8810 0.0213130140415313 \\ 
          8860 0.0212527730050430 \\ 
          8910 0.0211930135193832 \\ 
          8960 0.0211337839772056 \\ 
          9010 0.0210750379220051 \\ 
          9060 0.0210167689352988 \\ 
          9110 0.0209589757436071 \\ 
          9160 0.0209016886782784 \\ 
          9210 0.0208448371702644 \\ 
          9260 0.0207884519722210 \\ 
          9310 0.0207325478144340 \\ 
          9360 0.0206770618730575 \\ 
          9410 0.0206220397973040 \\ 
          9460 0.0205674477635234 \\ 
          9510 0.0205132847463468 \\ 
          9560 0.0204595188162571 \\ 
          9610 0.0204062279753431 \\ 
          9660 0.0203533275850964 \\ 
          9710 0.0203008258591935 \\ 
          9760 0.0202487388853213 \\ 
          9810 0.0201970549078456 \\ 
          9860 0.0201457482974589 \\ 
          9910 0.0200948388764249 \\ 
          9960 0.0200443214329036 \\ 
        };
        \addlegendentry{Upper Bound 
        \\
        (\autoref{thm:gain})};
        
        \hide{
        \addplot[
          color=black,
          solid,
          dash dot,
        ] table[row sep=crcr]{%
          31 0.385097848942280 \\
          73 0.241897763211972 \\
          127 0.180968523412823 \\
          179 0.151598660307930 \\
          233 0.132454524631041 \\
          283 0.119959709127033 \\
          353 0.107221214198948 \\
          419 0.0983045269904235 \\
          467 0.0930578466593114 \\
          547 0.0859158187708564 \\
          607 0.0815216610590769 \\
          661 0.0780940298779198 \\
          739 0.0738277277311686 \\
          811 0.0704527634419962 \\
          877 0.0677338251767658 \\
          947 0.0651683470382580 \\
          1019 0.0628118488213274 \\
          1087 0.0608058731141923 \\
          1153 0.0590319274805029 \\
          1229 0.0571696946066386 \\
          1297 0.0556447546130825 \\
          1381 0.0539193976238007 \\
          1453 0.0525616502429260 \\
          1523 0.0513353476506035 \\
          1597 0.0501279562121122 \\
          1663 0.0491200215853842 \\
          1741 0.0480037494986593 \\
          1823 0.0469085826411913 \\
          1901 0.0459334855610506 \\
          1993 0.0448579717262804 \\
          2063 0.0440884274234651 \\
          2131 0.0433775658778507 \\
          2221 0.0424875146415709 \\
          2293 0.0418136175596968 \\
          2371 0.0411185621540784 \\
          2437 0.0405567488429469 \\
          2539 0.0397320591477272 \\
          2621 0.0391043527469940 \\
          2689 0.0386057856830423 \\
          2749 0.0381813541064502 \\
          2833 0.0376099984588591 \\
          2909 0.0371145672406389 \\
          3001 0.0365402441636364 \\
          3083 0.0360502050705740 \\
          3187 0.0354561516947378 \\
          3259 0.0350616759380375 \\
          3343 0.0346176842966358 \\
          3433 0.0341602075153168 \\
          3517 0.0337491944159301 \\
          3581 0.0334458110325590 \\
          3659 0.0330868944885838 \\
          3733 0.0327568502760334 \\
          3823 0.0323684499846403 \\
          3911 0.0320017359361852 \\
          4001 0.0316392888108090 \\
          4073 0.0313580337549633 \\
          4153 0.0310541576079800 \\
          4241 0.0307298876393150 \\
          4327 0.0304226050038539 \\
          4421 0.0300970597932223 \\
          4507 0.0298081966470075 \\
          4591 0.0295339323821665 \\
          4663 0.0293047781339332 \\
          4759 0.0290073753272802 \\
          4861 0.0287010988121837 \\
          4943 0.0284617890800056 \\
          5009 0.0282734620832225 \\
          5099 0.0280225751587006 \\
          5189 0.0277782506385354 \\
          5281 0.0275349867727019 \\
          5393 0.0272472908613876 \\
          5449 0.0271067838703316 \\
          5527 0.0269146511333324 \\
          5641 0.0266410494766947 \\
          5701 0.0265003592930950 \\
          5801 0.0262707488387506 \\
          5861 0.0261358146463156 \\
          5953 0.0259328947811172 \\
          6067 0.0256878877523685 \\
          6143 0.0255283557364561 \\
          6229 0.0253513685902277 \\
          6311 0.0251859966241199 \\
          6373 0.0250630857073536 \\
          6481 0.0248532134135563 \\
          6577 0.0246710201340711 \\
          6679 0.0244817634412761 \\
          6763 0.0243291331405953 \\
          6841 0.0241899315062559 \\
          6947 0.0240045338903621 \\
          7001 0.0239117102814891 \\
          7109 0.0237292484032141 \\
          7211 0.0235607035798486 \\
          7307 0.0234053102049764 \\
          7417 0.0232309801682133 \\
          7507 0.0230912077486838 \\
          7573 0.0229902965394718 \\
          7649 0.0228757186269718 \\
          7727 0.0227598888488926 \\
          7841 0.0225937204361038 \\
          7927 0.0224707461212597 \\
          8039 0.0223135634358660 \\
          8117 0.0222060252505720 \\
          8221 0.0220650301888226 \\
          8293 0.0219689767902884 \\
          8389 0.0218428354494855 \\
          8513 0.0216830735207113 \\
          8599 0.0215743074612060 \\
          8681 0.0214721100785181 \\
          8747 0.0213909002481090 \\
          8837 0.0212816300366226 \\
          8933 0.0211669011763051 \\
          9013 0.0210726981553735 \\
          9127 0.0209406061351691 \\
          9203 0.0208539124152511 \\
          9293 0.0207526281479012 \\
          9391 0.0206440018771082 \\
          9461 0.0205674477635234 \\
          9539 0.0204831397437747 \\
          9643 0.0203723250784041 \\
          9739 0.0202716149540080 \\
          9817 0.0201908786422852 \\
          9901 0.0201050015926493 \\
          10009 0.0199961815840053 \\
          10103 0.0199028923785787 \\
          10181 0.0198264655442021 \\
          10273 0.0197374428284566 \\
          10357 0.0196571998520610 \\
          10463 0.0195573239030035 \\
          10589 0.0194405607625765 \\
          10663 0.0193729528042995 \\
          10753 0.0192916697790561 \\
          10861 0.0191954675746925 \\
          10957 0.0191111518890630 \\
          11069 0.0190141739464532 \\
          11159 0.0189373060544924 \\
          11257 0.0188546565721783 \\
          11351 0.0187763888264464 \\
          11447 0.0186974534048300 \\
          11549 0.0186146657601078 \\
          11677 0.0185123147470484 \\
          11779 0.0184319515713403 \\
          11839 0.0183851652321158 \\
          11939 0.0183079735261502 \\
          12037 0.0182332612451078 \\
          12113 0.0181759466761737 \\
          12227 0.0180909792973584 \\
          12301 0.0180364586337037 \\
          12409 0.0179577652989220 \\
          12491 0.0178986998079348 \\
          12569 0.0178430530398393 \\
          12647 0.0177879220798106 \\
          12743 0.0177207649882266 \\
          12841 0.0176529873872980 \\
          12941 0.0175846220506252 \\
          13009 0.0175385847535401 \\
          13121 0.0174635405432822 \\
          13217 0.0173999779062339 \\
          13313 0.0173371043078542 \\
          13417 0.0172697544764368 \\
          13513 0.0172082770355439 \\
          13627 0.0171361183576506 \\
          13709 0.0170847724595529 \\
          13789 0.0170351211076555 \\
          13883 0.0169773302585448 \\
          13997 0.0169080264133073 \\
          14083 0.0168563025787765 \\
          14207 0.0167825527409847 \\
          14327 0.0167120958276713 \\
          14423 0.0166563647731726 \\
          14533 0.0165931866301762 \\
          14621 0.0165431584800706 \\
          14713 0.0164913371364208 \\
          14771 0.0164589163628654 \\
          14867 0.0164056720182586 \\
          14951 0.0163595046577556 \\
          15077 0.0162909786384222 \\
          15161 0.0162457701469552 \\
          15263 0.0161913768203996 \\
          15329 0.0161564709043230 \\
          15413 0.0161123699784704 \\
          15511 0.0160613725138093 \\
          15619 0.0160057282783985 \\
          15683 0.0159730255773833 \\
          15787 0.0159203085776035 \\
          15887 0.0158701082613837 \\
          15973 0.0158273136936724 \\
          16073 0.0157779852351630 \\
          16187 0.0157223095021005 \\
          16273 0.0156806963569300 \\
          16411 0.0156146068276864 \\
          16487 0.0155785646529679 \\
          16607 0.0155221606536127 \\
          16693 0.0154821125803074 \\
          16823 0.0154221589376645 \\
          16921 0.0153774207309458 \\
          17011 0.0153366757448324 \\
          17099 0.0152971477481305 \\
          17203 0.0152508245570906 \\
          17321 0.0151987717625273 \\
          17393 0.0151672713574432 \\
          17483 0.0151281699063706 \\
          17579 0.0150867932136032 \\
          17681 0.0150432003462568 \\
          17789 0.0149974525886878 \\
          17903 0.0149496132207692 \\
          17977 0.0149188035278122 \\
          18061 0.0148840601481002 \\
          18149 0.0148479214151224 \\
          18251 0.0148063609222253 \\
          18329 0.0147748137937482 \\
          18433 0.0147330629311885 \\
          18521 0.0146980103735281 \\
          18661 0.0146427569326650 \\
          18757 0.0146052270175148 \\
          18911 0.0145456209268806 \\
          19013 0.0145065409820957 \\
          19139 0.0144586978622844 \\
          19231 0.0144240622687924 \\
          19333 0.0143859513850311 \\
          19427 0.0143510957521452 \\
          19483 0.0143304507350459 \\
          19577 0.0142959960007098 \\
          19709 0.0142480296009629 \\
          19801 0.0142148825706581 \\
          19913 0.0141748402049083 \\
          19993 0.0141464448071581 \\
          20089 0.0141125945362803 \\
          20161 0.0140873656970867 \\
          20269 0.0140497748540736 \\
          20359 0.0140186779719843 \\
          20477 0.0139782176060072 \\
          20563 0.0139489492641691 \\
          20707 0.0139003507846338 \\
          20773 0.0138782455942087 \\
          20899 0.0138363359937346 \\
          21001 0.0138026858537572 \\
          21089 0.0137738507929054 \\
          21179 0.0137445465221272 \\
          21283 0.0137109156371033 \\
          21391 0.0136762512382557 \\
          21493 0.0136437528958505 \\
          21569 0.0136196885020680 \\
          21649 0.0135944946173991 \\
          21757 0.0135607035953083 \\
          21851 0.0135314970841078 \\
          21961 0.0134975577061444 \\
          22051 0.0134699782696270 \\
          22129 0.0134462123459105 \\
          22247 0.0134104966727070 \\
          22343 0.0133816488728915 \\
          22447 0.0133506062115888 \\
          22549 0.0133203694168238 \\
          22651 0.0132903371390447 \\
          22739 0.0132645895147103 \\
          22853 0.0132314560638933 \\
          22961 0.0132002944231565 \\
          23039 0.0131779252253726 \\
          23117 0.0131556693639494 \\
          23209 0.0131295632376200 \\
          23327 0.0130963056869855 \\
          23459 0.0130594001501768 \\
          23563 0.0130305417928269 \\
          23633 0.0130112252721051 \\
          23747 0.0129799499825539 \\
          23831 0.0129570488074704 \\
          23911 0.0129353504801012 \\
          24019 0.0129062299111090 \\
          24097 0.0128853202803493 \\
          24179 0.0128634475713900 \\
          24317 0.0128268875341617 \\
          24419 0.0128000643657780 \\
          24527 0.0127718459770445 \\
          24671 0.0127345101408277 \\
          24781 0.0127062092718925 \\
          24889 0.0126786057350072 \\
          24979 0.0126557397001364 \\
          25111 0.0126224255250930 \\
          25189 0.0126028631082241 \\
          25307 0.0125734408004041 \\
          25409 0.0125481733215383 \\
          25523 0.0125201126744034 \\
          25609 0.0124990682478091 \\
          25703 0.0124761871458633 \\
          25801 0.0124524656695271 \\
          25919 0.0124240818266652 \\
          26003 0.0124039942328884 \\
          26113 0.0123778357930775 \\
          26209 0.0123551413417020 \\
          26309 0.0123316335258146 \\
          26407 0.0123087255655171 \\
          26501 0.0122868721506443 \\
          26641 0.0122545392153519 \\
          26713 0.0122380099304621 \\
          26813 0.0122151631553287 \\
          26893 0.0121969775777069 \\
          26993 0.0121743594153675 \\
          27091 0.0121523152528493 \\
          27239 0.0121192498779562 \\
          27337 0.0120975031903308 \\
        };
        \addlegendentry{Upper Bound~\cite{gooty2025precodingdesignlimitedfeedbackmiso}};
        }
        
        \addplot[
           color=black,
           solid,
           solid,
        ] table[row sep=crcr]{%
          10 0.330004772 \\ 
          60 0.131860269 \\ 
          110 0.0958628943 \\ 
          160 0.0794391346 \\ 
          210 0.0696297929 \\ 
          260 0.0622577291 \\ 
          310 0.0571788857 \\ 
          360 0.0525872618 \\ 
          410 0.0491795881 \\ 
          460 0.0468567556 \\ 
          510 0.0440553646 \\ 
          560 0.0423621190 \\ 
          610 0.0405834527 \\ 
          660 0.0390702143 \\ 
          710 0.0374516304 \\ 
          760 0.0363947160 \\ 
          810 0.0352476744 \\ 
          860 0.0341629835 \\ 
          910 0.0332452145 \\ 
          960 0.0322634696 \\ 
          1010 0.0315180302 \\ 
          1060 0.0307929470 \\ 
          1110 0.0300077520 \\ 
          1160 0.0294051524 \\ 
          1210 0.0287899102 \\ 
          1260 0.0280567697 \\ 
          1310 0.0276040518 \\ 
          1360 0.0271725821 \\ 
          1410 0.0265721561 \\ 
          1460 0.0261332002 \\ 
          1510 0.0257832394 \\ 
          1560 0.0253168227 \\ 
          1610 0.0249517558 \\ 
          1660 0.0245724525 \\ 
          1710 0.0241536288 \\ 
          1760 0.0237688552 \\ 
          1810 0.0235435253 \\ 
          1860 0.0232240725 \\ 
          1910 0.0229052992 \\ 
          1960 0.0225534556 \\ 
          2010 0.0223384658 \\ 
          2060 0.0220543378 \\ 
          2110 0.0218013941 \\ 
          2160 0.0215469841 \\ 
          2210 0.0212916517 \\ 
          2260 0.0210359015 \\ 
          2310 0.0208340574 \\ 
          2360 0.0205681928 \\ 
          2410 0.0203962854 \\ 
          2460 0.0201629435 \\ 
          2510 0.0199453306 \\ 
          2560 0.0197193076 \\ 
          2610 0.0195751393 \\ 
          2660 0.0194048888 \\ 
          2710 0.0192319094 \\ 
          2760 0.0190359061 \\ 
          2810 0.0188591341 \\ 
          2860 0.0187198139 \\ 
          2910 0.0185389041 \\ 
          2960 0.0183940883 \\ 
          3010 0.0182465531 \\ 
          3060 0.0180965987 \\ 
          3110 0.0179272529 \\ 
          3160 0.0178018259 \\ 
          3210 0.0176515054 \\ 
          3260 0.0175049284 \\ 
          3310 0.0173933472 \\ 
          3360 0.0172684441 \\ 
          3410 0.0171361459 \\ 
          3460 0.0170166637 \\ 
          3510 0.0168948337 \\ 
          3560 0.0167520844 \\ 
          3610 0.0166542115 \\ 
          3660 0.0165219214 \\ 
          3710 0.0164147515 \\ 
          3760 0.0163226412 \\ 
          3810 0.0161937967 \\ 
          3860 0.0161053364 \\ 
          3910 0.0160060219 \\ 
          3960 0.0158924772 \\ 
          4010 0.0158009599 \\ 
          4060 0.0156839876 \\ 
          4110 0.0156111687 \\ 
          4160 0.0154871973 \\ 
          4210 0.0154243920 \\ 
          4260 0.0153334827 \\ 
          4310 0.0152159580 \\ 
          4360 0.0151529102 \\ 
          4410 0.0150530548 \\ 
          4460 0.0149819245 \\ 
          4510 0.0148985736 \\ 
          4560 0.0148198459 \\ 
          4610 0.0147232031 \\ 
          4660 0.0146566267 \\ 
          4710 0.0145662048 \\ 
          4760 0.0144713817 \\ 
          4810 0.0144260953 \\ 
          4860 0.0143545866 \\ 
          4910 0.0142696071 \\ 
          4960 0.0142003889 \\ 
          5010 0.0141378065 \\ 
          5060 0.0140455197 \\ 
          5110 0.0139958880 \\ 
          5160 0.0139224699 \\ 
          5210 0.0138422449 \\ 
          5260 0.0137973516 \\ 
          5310 0.0137166805 \\ 
          5360 0.0136424744 \\ 
          5410 0.0136020030 \\ 
          5460 0.0135296325 \\ 
          5510 0.0134705896 \\ 
          5560 0.0134171333 \\ 
          5610 0.0133452845 \\ 
          5660 0.0132909556 \\ 
          5710 0.0132419190 \\ 
          5760 0.0131637141 \\ 
          5810 0.0131250880 \\ 
          5860 0.0130711635 \\ 
          5910 0.0130025068 \\ 
          5960 0.0129392442 \\ 
          6010 0.0129068472 \\ 
          6060 0.0128470805 \\ 
          6110 0.0127985900 \\ 
          6160 0.0127465067 \\ 
          6210 0.0126970976 \\ 
          6260 0.0126442196 \\ 
          6310 0.0125959858 \\ 
          6360 0.0125463273 \\ 
          6410 0.0124875177 \\ 
          6460 0.0124410483 \\ 
          6510 0.0123912924 \\ 
          6560 0.0123515263 \\ 
          6610 0.0122991150 \\ 
          6660 0.0122602271 \\ 
          6710 0.0122071475 \\ 
          6760 0.0121691228 \\ 
          6810 0.0121136607 \\ 
          6860 0.0120782595 \\ 
          6910 0.0120361957 \\ 
          6960 0.0119928452 \\ 
          7010 0.0119482559 \\ 
          7060 0.0119007922 \\ 
          7110 0.0118572179 \\ 
          7160 0.0118108265 \\ 
          7210 0.0117829060 \\ 
          7260 0.0117244797 \\ 
          7310 0.0116939702 \\ 
          7360 0.0116557822 \\ 
          7410 0.0116226672 \\ 
          7460 0.0115712197 \\ 
          7510 0.0115403530 \\ 
          7560 0.0115066876 \\ 
          7610 0.0114612664 \\ 
          7660 0.0114253060 \\ 
          7710 0.0113896819 \\ 
          7760 0.0113368848 \\ 
          7810 0.0113165261 \\ 
          7860 0.0112804728 \\ 
          7910 0.0112433420 \\ 
          7960 0.0112122026 \\ 
          8010 0.0111785291 \\ 
          8060 0.0111382497 \\ 
          8110 0.0111093423 \\ 
          8160 0.0110752271 \\ 
          8210 0.0110360500 \\ 
          8260 0.0110065970 \\ 
          8310 0.0109747368 \\ 
          8360 0.0109405347 \\ 
          8410 0.0109040594 \\ 
          8460 0.0108769330 \\ 
          8510 0.0108436343 \\ 
          8560 0.0108119018 \\ 
          8610 0.0107741942 \\ 
          8660 0.0107492637 \\ 
          8710 0.0107183501 \\ 
          8760 0.0106889214 \\ 
          8810 0.0106536891 \\ 
          8860 0.0106283844 \\ 
          8910 0.0105913709 \\ 
          8960 0.0105676853 \\ 
          9010 0.0105394770 \\ 
          9060 0.0105068555 \\ 
          9110 0.0104722398 \\ 
          9160 0.0104527638 \\ 
          9210 0.0104186782 \\ 
          9260 0.0103871621 \\ 
          9310 0.0103681562 \\ 
          9360 0.0103348888 \\ 
          9410 0.0103117837 \\ 
          9460 0.0102855702 \\ 
          9510 0.0102584770 \\ 
          9560 0.0102177147 \\ 
          9610 0.0102038651 \\ 
          9660 0.0101784641 \\ 
          9710 0.0101480275 \\ 
          9760 0.0101230401 \\ 
          9810 0.0101002945 \\ 
          9860 0.0100695315 \\ 
          9910 0.0100430934 \\ 
          9960 0.0100208813 \\ 
        };
        \addlegendentry{Random Codes 
        \\
        (\autoref{thm:random-chordal})};
        
        \end{axis}
        
    \end{tikzpicture}
    \caption{Bounds on the average covering radius of an $(n, n - 1)_{p}$ CRS code for $p = n + o(n)$}
    \label{fig:asymptotic-bounds-comparison}
\end{figure}

\section{Conclusion and Future Directions}
\label{sec:conclusion}

In this work, we developed efficient covering algorithms for both block and subspace codes based on RS codes. Central to our approach is a puncturing-based covering algorithm that leverages existing decoding techniques and enables efficient identification of codewords within the covering radius of the code. For block codes, we showed that when combined with the Guruswami–Sudan list decoder, the proposed algorithm achieves an average covering radius that closely approaches the optimal performance of the MAP decoder, despite operating with polynomial-time complexity.

From an information-theoretic perspective, we introduced the notion of the average covering radius for both block and subspace codes as a meaningful measure of average distortion in quantization problems. Using tools from one-shot compression and rate-distortion theory, we derived explicit random-coding bounds on the average covering radius in both Hamming and Grassmann spaces. These bounds provide fundamental benchmarks against which structured code constructions can be evaluated. In particular, our numerical results demonstrated that the RS-based constructions outperform random codebooks in various fixed block-length scenarios. The analysis provides new insights into the role of algebraic structures in covering problems and applications.

Several directions for future work naturally emerge from this study. First, it would be of interest to extend the present analysis to CP/CRS codes over arbitrary finite fields, as well as to more general Grassmannian code constructions. Second, the performance of the CRS covering algorithm may be further improved by refining the mapping step in \autoref{step:CRScover:psi} of \autoref{alg:CP-cover}, potentially through more deterministic or structure-aware choices. Additionally, the upper bound in \autoref{thm:gain} could be tightened by replacing Jensen’s inequality with sharper concentration bounds or by incorporating additional sources of randomness into the analysis.

A notable limitation of CRS codes is the constant-modulus constraint on codeword coordinates. This, while advantageous for practical wireless implementations, restricts geometric flexibility and likely contributes to their inferior performance relative to random codes in short-length regimes. This observation motivates the exploration of multi-level CRS codes with variable coordinate amplitudes, and a systematic study of their impact on average covering performance.

Finally, many of the concepts and techniques developed in this work extend naturally beyond the Hamming and Grassmann spaces. In particular, investigating average covering properties and efficient covering algorithms for codes in other metric spaces, such as rank-metric codes, is another promising avenue for future research.

\bibliographystyle{IEEEtran}
\bibliography{references}

@BOOK{ecc,
    url = {https://cse.buffalo.edu/faculty/atri/courses/coding-theory/book/},
    author = {Guruswami, Venkatesan and Rudra, Atri and Sudan, Madhu},
    title = {Essential Coding Theory},
    year = {2022},
    publisher = {online}
}

@ARTICLE{Hessam22,
    author = {Soleymani, Mahdi and Mahdavifar, Hessam},
    journal = {IEEE Transactions on Information Theory}, 
    title = {Analog Subspace Coding: A New Approach to Coding for Non-Coherent Wireless Networks}, 
    year = {2022},
    volume = {68},
    number = {4},
    pages = {2349-2364},
    publisher = {IEEE},
    doi = {10.1109/TIT.2021.3133071},
    URL = {https://doi.org/10.1109/TIT.2021.3133071},
    eprint = {https://doi.org/10.1109/TIT.2021.3133071}
}

@ARTICLE{Conway96,
    author = {John H. Conway and Ronald H. Hardin and Neil J. A. Sloane},
    title = {{Packing Lines, Planes, etc.: Packings in Grassmannian Spaces}},
    journal = {Experimental Mathematics},
    volume = {5},
    number = {2},
    pages = {139-159},
    year = {1996},
    publisher = {Taylor & Francis},
    doi = {10.1080/10586458.1996.10504585},
    URL = {https://doi.org/10.1080/10586458.1996.10504585},
    eprint = {https://doi.org/10.1080/10586458.1996.10504585}
}

@INPROCEEDINGS{Hessam21,
    author = {Soleymani, Mahdi and Mahdavifar, Hessam},
    booktitle = {2021 IEEE International Symposium on Information Theory (ISIT)}, 
    title = {New Packings in {Grassmannian} Space}, 
    year = {2021},
    volume = {},
    number = {},
    pages = {807-812},
    doi = {10.1109/ISIT45174.2021.9517770}
}

@ARTICLE{Guruswami06,
    author = {Venkatesan Guruswami},
    title = {Algorithmic Results in List Decoding},
    journal = {Foundations and Trends in Theoretical Computer Sciences},
    volume = {2},
    number = {2},
    pages = {107-195},
    year = {2006},
    publisher = {},
    doi = {10.1561/0400000007},
    URL = {https://www.cs.cmu.edu/~venkatg/pubs/papers/listdecoding-NOW.pdf},
    eprint = {https://www.cs.cmu.edu/~venkatg/pubs/papers/listdecoding-NOW.pdf}
}

@article{Roth00,
    author = {Roth, R.M. and Ruckenstein, G.},
    journal = {IEEE Transactions on Information Theory}, 
    title = {{Efficient decoding of Reed--Solomon codes beyond half the minimum distance}}, 
    year = {2000},
    volume = {46},
    number = {1},
    pages = {246-257},
    doi = {10.1109/18.817522}
}

@article{McEliece03,
    author = {McEliece, R. J.},
    journal = {IPN Progress Report 42-153}, 
    title = {{The Guruswami–Sudan Decoding Algorithm for Reed–Solomon Codes}}, 
    year = {2003},
    url = {https://tmo.jpl.nasa.gov/progress_report/42-153/153F.pdf},
    eprint = {https://tmo.jpl.nasa.gov/progress_report/42-153/153F.pdf}
}

@inproceedings{McEliece03-1,
    title = {{On the average list size for the Guruswami--Sudan decoder}},
    author = {McEliece, Robert J},
    booktitle = {7th International Symposium on Communications Theory and Applications (ISCTA)},
    volume = {194},
    year = {2003}
}

@InProceedings{Shokrollahi00,
    author = "Gao, Shuhong
    and Shokrollahi, M. Amin",
    editor = "Joyner, David",
    title = "{Computing Roots of Polynomials over Function Fields of Curves}",
    booktitle = "Coding Theory and Cryptography",
    year = "2000",
    publisher = "Springer Berlin Heidelberg",
    address = "Berlin, Heidelberg",
    pages = "214--228",
    isbn = "978-3-642-59663-6"
}

@article{KK,
    title={Coding for errors and erasures in random network coding},
    author={Koetter, Ralf and Kschischang, Frank R},
    journal={IEEE Transactions on Information Theory},
    volume={54},
    number={8},
    pages={3579--3591},
    year={2008},
    publisher={IEEE}
}

@article{barg2002bounds,
    title={Bounds on packings of spheres in the {Grassmann} manifold},
    author={Barg, Alexander and Nogin, D {\relax Yu}},
    journal={IEEE Transactions on Information Theory},
    volume={48},
    number={9},
    pages={2450--2454},
    year={2002},
    publisher={IEEE}
}

@book{Huffman03, 
    place={Cambridge}, 
    title={Fundamentals of Error-Correcting Codes}, 
    publisher={Cambridge University Press}, 
    author={Huffman, W. Cary and Pless, Vera}, 
    year={2003}
}

@book{Macwilliams77,
    title={The Theory of Error-Correcting Codes},
    author={MacWilliams, F.J. and Sloane, N.J.A.},
    isbn={9780444850102},
    lccn={76041296},
    series={Mathematical Library},
    year={1977},
    publisher={North-Holland Publishing Company}
}

@book{Cohen97,
    title={Covering Codes},
    author={Cohen, G{\'e}rard and Honkala, Iiro and Litsyn, Simon and Lobstein, Antoine},
    year={1997},
    publisher={Elsevier}
}

@article{Brualdi98,
    title={{Covering Radius}},
    author={Brualdi, RA and Litsyn, S and Pless, VS},
    journal={Handbook of Coding Theory},
    volume={1},
    pages={755--826},
    year={1998},
    publisher={Elsevier Amsterdam}
}

@ARTICLE{Cohen86,
    author={Cohen, G. and Lobstein, A. and Sloane, N.},
    journal={IEEE Transactions on Information Theory}, 
    title={{Further Results on the Covering Radius of Codes}}, 
    year={1986},
    volume={32},
    number={5},
    pages={680-694},
    doi={10.1109/TIT.1986.1057227}
}

@ARTICLE{Cohen85,
    author={Cohen, G. and Karpovsky, M. and Mattson, H. and Schatz, J.},
    journal={IEEE Transactions on Information Theory}, 
    title={{Covering Radius---Survey and Recent Results}}, 
    year={1985},
    volume={31},
    number={3},
    pages={328-343},
    doi={10.1109/TIT.1985.1057043}
}

@InProceedings{vanLint88,
    author="van Lint, J. H.",
    editor="Mora, Teo",
    title="Recent results on covering problems",
    booktitle="Applied Algebra, Algebraic Algorithms and Error-Correcting Codes",
    year="1989",
    publisher="Springer Berlin Heidelberg",
    address="Berlin, Heidelberg",
    pages="7--21",
    isbn="978-3-540-46152-4"
}

@inproceedings{gooty2025precodingdesignlimitedfeedbackmiso,
    title={{Precoding Design for Limited-Feedback MISO Systems via Character-Polynomial Codes}}, 
    author={Siva Aditya Gooty and Samin Riasat and Hessam Mahdavifar and Robert W. {Heath Jr}},
    year={2025},
    note = {{(extended version)}}, 
    url = {https://arxiv.org/abs/2501.06178}, 
    eprint={2501.06178},
    archivePrefix={arXiv},
    primaryClass={eess.SP},
    booktitle = {IEEE International Conference on Communications (ICC): Communication Theory Symposium}
}

@inproceedings{riasat2024decodinganalogsubspacecodes,
    title={{Decoding Analog Subspace Codes: Algorithms for Character-Polynomial Codes}}, 
    author={Samin Riasat and Hessam Mahdavifar},
    year={2024},
    eprint={2407.03606},
    archivePrefix={arXiv},
    primaryClass={cs.IT},
    url={https://arxiv.org/abs/2407.03606}, 
    booktitle = {IEEE International Symposium on Information Theory (ISIT)}, 
}

@inproceedings{riasat2025covering,
    title={{Efficient Covering Using Reed--Solomon Codes}}, 
    author={Samin Riasat and Hessam Mahdavifar},
    year={2025},
    eprint={2502.01984},
    archivePrefix={arXiv},
    primaryClass={cs.IT},
    url={https://arxiv.org/abs/2502.01984}, 
    booktitle = {IEEE International Symposium on Information Theory (ISIT)}, 
}

@article{torquato10,
    title = {Reformulation of the covering and quantizer problems as ground states of interacting particles},
    author = {Torquato, S.},
    journal = {Phys. Rev. E},
    volume = {82},
    issue = {5},
    pages = {056109},
    numpages = {22},
    year = {2010},
    month = {Nov},
    publisher = {American Physical Society},
    doi = {10.1103/PhysRevE.82.056109},
    url = {https://link.aps.org/doi/10.1103/PhysRevE.82.056109}
}

@misc{toth22,
    title={Packing and covering in higher dimensions}, 
    author={Gábor Fejes Tóth},
    year={2022},
    eprint={2202.11358},
    archivePrefix={arXiv},
    primaryClass={math.MG},
    url={https://arxiv.org/abs/2202.11358}, 
}

@INPROCEEDINGS{elkayam20,
    author={Elkayam, Nir and Feder, Meir},
    booktitle={2020 IEEE International Symposium on Information Theory (ISIT)}, 
    title={One shot approach to lossy source coding under average distortion constraints}, 
    year={2020},
    volume={},
    number={},
    pages={2389-2393},
    doi={10.1109/ISIT44484.2020.9173943}, 
    url={https://ieeexplore.ieee.org/document/9173943}
}

@article{kus88,
    doi = {10.1088/0305-4470/21/22/006},
    url = {https://doi.org/10.1088/0305-4470/21/22/006},
    year = {1988},
    month = {nov},
    publisher = {},
    volume = {21},
    number = {22},
    pages = {L1073},
    author = {M Kus and J Mostowski and F Haake},
    title = {Universality of eigenvector statistics of kicked tops of different symmetries},
    journal = {Journal of Physics A: Mathematical and General}
}

@misc{gooty2025efficientdecoderssensingsubspace,
      title={Efficient Decoders for Sensing Subspace Code}, 
      author={Siva Aditya Gooty and Hessam Mahdavifar},
      year={2025},
      eprint={2512.05028},
      archivePrefix={arXiv},
      primaryClass={eess.SP},
      url={https://arxiv.org/abs/2512.05028}, 
}

@ARTICLE{McLoughlin84,
    author={McLoughlin, A.},
    journal={IEEE Transactions on Information Theory}, 
    title={The complexity of computing the covering radius of a code}, 
    year={1984},
    volume={30},
    number={6},
    pages={800-804},
    doi={10.1109/TIT.1984.1056978}
}

@article{Dur94,
    title = {On the covering radius of {Reed--Solomon} codes},
    journal = {Discrete Mathematics},
    volume = {126},
    number = {1},
    pages = {99-105},
    year = {1994},
    issn = {0012-365X},
    doi = {https://doi.org/10.1016/0012-365X(94)90256-9},
    author = {Arne Dür}
}

@article{Dur91,
    title = {The decoding of extended {Reed--Solomon} codes},
    journal = {Discrete Mathematics},
    volume = {90},
    number = {1},
    pages = {21-40},
    year = {1991},
    issn = {0012-365X},
    doi = {https://doi.org/10.1016/0012-365X(91)90093-H},
    author = {Arne Dür}
}

@misc{qian2022coveringgrassmanniancodesbounds,
    title={{Covering Grassmannian Codes: Bounds and Constructions}}, 
    author={Bingchen Qian and Xin Wang and Chengfei Xie and Gennian Ge},
    year={2022},
    eprint={2207.09277},
    archivePrefix={arXiv},
    primaryClass={cs.IT},
    url={https://arxiv.org/abs/2207.09277}, 
}

@misc{blackburn2011asymptoticbehaviorgrassmanniancodes,
    title={{The asymptotic behavior of Grassmannian codes}}, 
    author={Simon R. Blackburn and Tuvi Etzion},
    year={2011},
    eprint={1111.2713},
    archivePrefix={arXiv},
    primaryClass={cs.DM},
    url={https://arxiv.org/abs/1111.2713}, 
}

@article{reed1960polynomial,
    title={Polynomial codes over certain finite fields},
    author={Reed, Irving S and Solomon, Gustave},
    journal={Journal of the society for industrial and applied mathematics},
    volume={8},
    number={2},
    pages={300--304},
    year={1960},
    publisher={SIAM}
}

@inproceedings{mahdavifar2024subspace,
    title={Subspace Coding for Spatial Sensing},
    author={Mahdavifar, Hessam and Rajam{\"a}ki, Robin and Pal, Piya},
    booktitle={2024 IEEE International Symposium on Information Theory (ISIT)},
    pages={2394--2399},
    year={2024},
    organization={IEEE}
}

@article{shi2022covering,
    title={Covering radius of {Melas} codes},
    author={Shi, Minjia and Helleseth, Tor and {\"O}zbudak, Ferruh and Sol{\'e}, Patrick},
    journal={IEEE Transactions on Information Theory},
    volume={68},
    number={7},
    pages={4354--4364},
    year={2022},
    publisher={IEEE}
}

@article{shi2023covering,
    title={Covering radius of generalized {Zetterberg} type codes over finite fields of odd characteristic},
    author={Shi, Minjia and Helleseth, Tor and {\"O}zbudak, Ferruh},
    journal={IEEE Transactions on Information Theory},
    volume={69},
    number={11},
    pages={7025--7048},
    year={2023},
    publisher={IEEE}
}

@article{shi2025covering,
    title={The covering radius of {Euclidean} codes},
    author={Shi, Minjia and Xu, Xingxing and Sol{\'e}, Patrick},
    journal={Journal of Applied Mathematics and Computing},
    year={2025}
}

@article{shi2025determining,
    title={Determining the covering radius of all generalized {Zetterberg} codes in odd characteristic},
    author={Shi, Minjia and Li, Shitao and Helleseth, Tor and {\"O}zbudak, Ferruh},
    journal={IEEE Transactions on Information Theory},
    year={2025},
    publisher={IEEE}
}

@article{zhu2025dual,
    title={Dual and Covering Radii of Extended Algebraic Geometry Codes},
    author={Zhu, Yunlong and Zhao, Chang-An},
    journal={arXiv preprint arXiv:2509.21773},
    year={2025}
}

@article{Gautschi1959SomeEI,
    title={Some Elementary Inequalities Relating to the Gamma and Incomplete Gamma Function},
    author={Walter Gautschi},
    journal={Journal of Mathematics and Physics},
    year={1959},
    volume={38},
    pages={77-81},
}

@article{BHP01,
    author = {Baker, R. C. and Harman, G. and Pintz, J.},
    title = {{The Difference Between Consecutive Primes, II}},
    journal = {Proceedings of the London Mathematical Society},
    volume = {83},
    number = {3},
    pages = {532-562},
    doi = {https://doi.org/10.1112/plms/83.3.532},
    url = {https://londmathsoc.onlinelibrary.wiley.com/doi/abs/10.1112/plms/83.3.532},
    eprint = {https://londmathsoc.onlinelibrary.wiley.com/doi/pdf/10.1112/plms/83.3.532},
    year = {2001}
}

@misc{ordentlich2026voronoisphericalcdflattices,
    title={{The Voronoi Spherical CDF for Lattices and Linear Codes: New Bounds for Quantization and Coding}}, 
    author={Or Ordentlich},
    year={2026},
    eprint={2506.19791},
    archivePrefix={arXiv},
    primaryClass={cs.IT},
    url={https://arxiv.org/abs/2506.19791}, 
}

\end{document}